\documentclass[11pt,english]{article}
\usepackage{ae,aecompl}
\usepackage[T1]{fontenc}
\usepackage[latin9]{inputenc}
\usepackage{geometry}
\usepackage{color}
\usepackage{babel}
\usepackage{amsmath}
\usepackage{amsthm}
\usepackage{appendix}
\usepackage{amssymb}
\usepackage{graphicx}
\usepackage{setspace}
\usepackage[authoryear]{natbib}
\usepackage[unicode=true,pdfusetitle,
 bookmarks=true,bookmarksnumbered=false,bookmarksopen=false,
 breaklinks=false,pdfborder={0 0 0},pdfborderstyle={},backref=false,colorlinks=true]
 {hyperref}
\hypersetup{
 pdfpagelayout=OneColumn, pdfnewwindow=true, pdfstartview=XYZ, plainpages=false, urlcolor=[rgb]{0.0430 ,0, 0.5}, linkcolor=[rgb]{0.0430 ,0, 0.5}, citecolor=[rgb]{0.0430 ,0, 0.5}, hypertexnames=false}

\makeatletter

\newcommand{\Prob}{\mathbb{P}}
\newcommand{\abs}[1]{\left\lvert{#1}\right\rvert} 

\theoremstyle{definition}
\newtheorem{defn}{\protect\definitionname}
\theoremstyle{plain}
\newtheorem{prop}{\protect\propositionname}
\theoremstyle{definition}
 \newtheorem{example}{\protect\examplename}
\theoremstyle{plain}

\newtheorem{theorem}{\protect\theoremname}

\makeatother

\providecommand{\definitionname}{Definition}
\providecommand{\examplename}{Example}
\providecommand{\theoremname}{Theorem}
\providecommand{\lemmaname}{Lemma}
\providecommand{\propositionname}{Proposition}
\providecommand{\corollaryname}{Corollary}

\title{Identifying Wisdom (of the Crowd): A Regression Approach}
\author{Jonathan Libgober \thanks{University of Southern California. Email: \texttt{\protect\href{mailto:libgober\%40usc.edu}{libgober@usc.edu}}. This is a revision of ``Hypothetical Beliefs Identify Information.'' The main results are largely similar, although the model has been rewritten to follow the literature on opinion aggregation and crowd wisdom. I thank an associate editor for suggesting this reframing and other comments, as well as three referees for insightful suggestions. I thank Juan Carrillo,  
Xiaosheng Mu, Mallesh Pai, and Larry Samuelson for helpful comments, as well as the hospitality of the Cowles Foundation and Yale University which hosted me during part of this research.}}
\date{January 29, 2023}

\begin{document}

\maketitle

\begin{abstract} 
\noindent Experts in a population hold (a) beliefs over a state (call these \emph{state beliefs}), as well as (b) beliefs over the distribution of beliefs in the population (call these \emph{hypothetical beliefs}). If these are generated via updating a common prior using a fixed information structure, then the information structure can (generically) be derived by regressing hypothetical beliefs on state beliefs, provided there are at least as many signals as states. In addition, the prior solves an eigenvector equation derived from a matrix determined by the state beliefs and the hypothetical beliefs. Thus, the ex-ante informational environment (i.e., how signals are generated) can be determined using ex-post data (i.e., the beliefs in the population). I discuss implications of this finding, as well as what is identified when there are more states than signals. 
\end{abstract}

\newpage

\section{Introduction}

Consider a population of experts independently forming opinions over some possible events---for instance,  meteorologists predicting the weather, medical professionals seeking to diagnose a patient, or consultants attempting to determine the profitability of an investment. An outside analyst (e.g., an econometrician) knows that every time a new prediction confronts the group, the opinions of each member are formed by updating beliefs based on a fixed common prior using independently drawn signals, the latter of which depend on the truth according to some fixed information structure. 

This paper is interested in the question of how to determine the content of the experts' information in situations like the ones described above. Central to my approach is an assumption that the analyst can observe not only the beliefs of each group member regarding the relevant state, but also their beliefs regarding the distribution of their peers' beliefs. In this paper, I refer to this latter object as the ``hypothetical beliefs,'' as they are equivalently described as the probability distribution an expert would assign, given their (ex-post) belief, to each possible belief of a hypothetical expert, were this expert drawn uniformly at random from the population; insofar as experts view themselves as exchangeable, this corresponds to the beliefs regarding the hypothetical possibility of having received a different signal. Note that this will typically differ from the true distribution of signals in the population whenever the expert does not know the true state.

The analyst is interested in learning the informational environment which determines the population's beliefs. This consists of two components: First, the prior over the set of possible states (e.g., weather events, diseases, success outcomes); and second, the  Blackwell experiment---that is, the function which maps each state of the world to a distribution over the set of possible signals---which the experts use when forming their beliefs. The problem is that these states are not directly observed, and the experts themselves may not be able to describe how states map into a distribution over signals directly. Note that, if the analyst were to observe the \emph{true} distribution over beliefs in the group, then knowledge of the Blackwell experiment could allow the analyst to learn the state itself, simply by matching the observed belief distribution to the predicted one given the state. This question, on how to infer the true state of the world given the group's ex-post beliefs, is the main focus of the literature on \emph{crowd wisdom}, discussed in more depth below. The basic framework I develop is in fact identical to one from \cite{PSM2017} (see also \cite{PM2022}). This literature shares the assumption that it is prohibitive to elicit the distribution over signals conditional on the state itself. This could reflect limited memory among the group, or simply that experts only form beliefs ``as-if'' updating from a common prior, without knowledge of the prior or Blackwell experiment itself.\footnote{An example from \cite{PSM2017} is the question of whether Philadelphia is the capital of Pennsylvania. The questions (i) ``what probability do you assign to Philadelphia being the capital of Pennsylvania?'' and (ii) ``given a probability $p$, what probability would you assign to a randomly selected person believing that Philadelphia is the capital of Pennsylvania with probability $p$?'' are fairly straightforward to formulate; it may be harder to answer the question ``if Philadelphia were the capital of Pennsylvania, what probability would you assign to someone believing Philadelphia were the capital with probability $p$?''}

\cite{PSM2017} point out that the ex-post beliefs of the population may fail to reflect the truth if instead they simply reflect a prior which heavily weights one state (see also \cite{ABS2020}). If the prior is such that one state, say $\theta^{*}$, is weighted more heavily than others and the signal observed is sufficiently weak, then the experts will always simply believe $\theta^{*}$ is more likely, even when it is not the true state. The issues related to inference of the information structure are slightly different, but again stem from the observation that the prior is confounding---a given set of beliefs over the state could emerge given \emph{any} prior  in the interior of the convex hull of these beliefs. In turn, different priors may very well correspond to different Blackwell experiments.\footnote{For instance, suppose there are two possible states, $\{1,2\}$ and two possible signals, $\{1,2\}$. Suppose that following signal 1, state 1 has posterior probability 3/4, and under signal 2, state 1 has posterior probability 1/4. If each state is ex-ante equally likely, then these ex-post beliefs are consistent with the signal being equal to the state with probability 3/4. But it may be that the prior probability of state 1 is higher, with these same ex-post beliefs being induced. The same information structure would not yield the same ex-post beliefs with the different prior.} Thus, beliefs over the state are insufficient to pin down the Blackwell experiment. 

My main result is that the Blackwell experiment in the experts' problem can be determined as the outcome of a strikingly simple regression procedure I describe. The hypothetical beliefs of the population play a crucial role in this procedure. The simplest case is when there are at least as many signals as states. In this case, the procedure is to simply regress a hypothetical belief vector (i.e., the probability assigned to possibly having observed some fixed, particular signal) on the matrix of beliefs. While specifying these objects appropriately requires some care, this regression delivers the information structure which generates the signals. 

The prior the experts update from also has a geometric interpretation as an eigenvector of a matrix which comes out of analyzing the martingale condition on beliefs; in particular, the prior is the unique eigenvector corresponding to eigenvalue 1, normalized so that the entries sum to 1. Under mild assumptions, one can back out both the information structure and the prior from the possible beliefs over the states, together with the hypothetical beliefs. I also discuss what it means for these conditions to fail, and what can be done when they do. 

If there are more signals than states, then the problem above cannot be solved via the same method. Again, the issue is familiar from linear regression, where this takes the form of an identification problem that emerges if there are more explanatory variables than observations. It turns out that one proposal from statistics for how to address this problem can allow the analyst to learn something in my situation as well---I describe a regularization process which essentially allows us to perform the inversion step required by linear regression. The process is known as ``ridge regression.'' The idea is to add a small perturbation to the singular matrix to avoid the invertibility issue that arises with the identification failure. I show that this procedure identifies a subspace on which the information structure must lie, up to a vector in the null space in the matrix of beliefs. In some cases this restriction may allow the analyst to determine the information structure. But perhaps more surprisingly, I show that even if the procedure does not identify the information structure, it nevertheless does induce the same eigenvector interpretation of the prior (even with the ``incorrect'' information structure), allowing it to be recovered. 

My analysis therefore shows how intimately related hypothetical beliefs are to the underlying information structure. In the course of the analysis, I discuss the geometric relationship between the ex-ante determinants of the experts' environment (i.e., the prior and Blackwell experiment) and the ex-post data (i.e., the beliefs observed by the analyst). I briefly mention the two most notable observations along these lines: First, I consider the dimensionality of the set of hypothetical belief which can emerge, given a fixed matrix of beliefs over the state. If there are more signals than states, then this is generally smaller than the set of hypothetical belief matrices which can emerge based only on the requirement that probabilities are non-zero and sum to one, sometimes significantly so. I interpret this as a word of caution, as ``most'' possible hypothetical belief matrices will not correspond to any informational environment. Second, I consider relaxing the assumption that all subjects use a common prior, by allowing priors to depend on signal realizations themselves (even holding fixed the relevant Blackwell experiment). In this case, there are no restrictions at all---any $Q$ can be rationalized. 

The contribution of this paper is in showing how the regression approach allows for the analyst to identify the informational environment in an intuitive way. I note, while the focus of \cite{PSM2017} is on inferring the state, the procedure they identify can also be used to infer the relevant Blackwell experiment. Briefly, the core of their approach is to first determine the ex-ante probability that each signal is observed, whereas regression first determines the Blackwell experiment. I contrast these approaches and ultimately conclude that they are complementary.\footnote{In fact, there is a precise sense in which these approaches are \emph{dual} to one another. Specifically, their argument uses an eigenvector characterization of the ex-ante probability each signal is observed. By contrast, I discuss how regression lends itself to an eigenvector characterization of the prior over states. Thus, I identify a duality between ``priors over signals'' and ``priors over states'' that emerges in the course of the analysis.} While in many cases these are two paths to the same ends, there are instances where regression enables the analyst to learn something in cases where the ex-ante signal distribution could not be identified. On the other hand, there are cases where the \cite{PSM2017} approach is more direct and simpler to implement than regression. Nevertheless, clarifying and revealing the mathematical structure underlying these objects may prove useful in subsequent analysis of this and other closely related environments.  

\section{Preliminaries}  \label{sect:model}

A continuum population of mass one forms beliefs over a finite state space $\Theta$. The belief that each individual holds over $\Theta$ is summarized by a \emph{belief type}. Let $S$ denote the set of belief types which are possible in the population; I assume throughout that $\abs{S} < \infty$. Write $b_{s,\theta}$ for the probability an individual with belief type $s \in S$ assigns to the state being $\theta \in \Theta$; I call such a belief a \emph{state belief} when necessary to avoid confusion with hypothetical beliefs (introduced below). I denote the $\abs{S}$-by-$\abs{\Theta}$ matrix of beliefs (over the state) by $B=(b_{s,\theta})_{s \in S, \theta \in \Theta}$. Here, rows are belief types, and columns are states $\theta \in \Theta$ over which the beliefs are formed. I refer to the matrix $B$ as the state belief matrix. I assume that $B$ does not have a zero vector for any column, that $B$ is full rank.

Each individual is able to form not only a belief over the state, but also a conjecture of the distribution of belief types in the population. I call such beliefs \emph{hypothetical beliefs}. Formally, let $q_{s,\tilde{s}}$ denote the probability an individual $i$ with belief type $s$ assigns to an individual $j$ being of belief type $\tilde{s}$, when $j$ is drawn uniformly at random from the population. I let $Q=(q_{s, \tilde{s}})_{s \in S, \tilde{s} \in S}$ denote the corresponding $\abs{S}$-by-$\abs{S}$ matrix, where rows index belief types of individual $i$, and the columns index the belief types the randomly drawn individual $j$ observes in this exercise (so that $Q$ is row-stochastic). Equivalently, the row corresponding to $s$ gives the expected distribution over belief types in the population held by an individual of belief type $s$. I refer to the matrix $Q$ as the  hypothetical belief matrix. 

I refer to the combination of $B$ and $Q$ as the \emph{belief landscape}. In this paper, I take a belief landscape as a primitive. I will call a belief landscape \emph{plausible} if $B$ and $Q$ are row stochastic and non-negative. If either of these conditions are violated, then it is immediately apparent that it is not possible to interpret $B$ and $Q$ as I have so far, namely reflecting Bayesian probabilities.

In principle, a belief landscape could be arbitrary, but I am interested in cases where it is generated by the population consisting of Bayesian decision-makers who all receive conditionally IID signals from a fixed Blackwell experiment. A \emph{Blackwell experiment} or \emph{information structure} is a function $\mathcal{I} : \Theta \rightarrow  \Delta(S)$. Let $\mathcal{I}(\theta)[s]$ denote the probability that the decision-maker observes $s$ in state $\theta$. Given a prior $p \in \Delta(\Theta)$, a decision-maker who observes a signal $s$ from a Blackwell experiment can update beliefs over each state $\theta \in \Theta$ via Bayes rule; that is, if the decision-maker's belief type $s$ is determined by the information structure $\mathcal{I}$, then: 

\begin{equation} 
b_{s, \theta} = \frac{p(\theta) \mathcal{I}(\theta)[s]}{\sum_{\tilde{\theta} \in \Theta} p(\tilde{\theta}) \mathcal{I}(\tilde{\theta})[s]}. \label{eq:bayes}
\end{equation}

\noindent In addition, $(q_{s,s_{1}}, \ldots, q_{s,s_{n}})$ is pinned down as well. If all individuals have access to the same Blackwell experiment, then $\mathcal{I}(\theta)[\tilde{s}]$ is the fraction of the population that obtains signal $\tilde{s}$ in state $\theta$ (and thus the probability that a randomly drawn individual has belief type $\tilde{s}$). Thus, if an individual of belief type $s$ holds beliefs $(b_{s,\theta})_{\theta \in \Theta}$, then the law of total probability implies: 

\begin{equation} 
q_{s,\tilde{s}}= \sum_{\theta} \mathcal{I}(\theta)[\tilde{s}] b_{s,\theta} \label{eq:totalprob}
\end{equation}

\noindent The main question of this paper is whether one can take the population's beliefs to be generated by all individuals having access to the same $\mathcal{I}$ and updating from the same prior $p$. 

\begin{defn}  \label{def:informationenviron}
A belief landscape $(B,Q)$ is generated by an informational environment $(\mathcal{I}, p)$ if $b_{s,\theta}$ and $q_{s,\tilde{s}}$ can be derived using prior $p$ and taking the belief type $s$ to be drawn according to $\mathcal{I}$ using (\ref{eq:bayes}) and (\ref{eq:totalprob}).
\end{defn}

\noindent I am interested in identifying the informational environment; however, as stated in the introduction, were the analyst to observe the fraction of the population comprising each belief type, then the state itself could be inferred as well by comparing this fraction of population with belief type $s$ to $\mathcal{I}(\theta)[s]$. This question is often the primary focus of the crowd wisdom literature (reviewed more completely below), with the identification of information being secondary, but I return to this observation when I compare my approach to others from this literature.

\begin{example} \label{ex:mainillustration}
I walk through a simple example to illustrate the key definitions from the previous section. Suppose $\Theta= \{\theta_{1}, \theta_{2}, \theta_{3}\}$, and suppose there is an initial prior over $\Theta$ that assigns probability $p(\theta_{i})$ to state $\theta_{i}$. Consider the following information structure: With probability $\varepsilon \in (0,1)$, a ``null signal'' is drawn. With probability $1-\varepsilon$, the state is observed. Using the above formalism, one can write the state belief matrix as follows:

\begin{equation*} 
B= \begin{pmatrix} p(\theta_{1}) & p(\theta_{2}) & p(\theta_{3}) \\ 1 & 0 & 0 \\ 0  & 1 & 0 \\ 0& 0 & 1 \end{pmatrix}.
\end{equation*} 

\noindent Notice that, consistent with the definitions above, the rows refer to different signals an individual might observe, and the columns refer to different states. 

The hypothetical belief matrix corresponding to this information structure is therefore: 
\begin{equation*} 
Q= \begin{pmatrix} 
\varepsilon & (1- \varepsilon) p(\theta_{1}) & (1- \varepsilon) p(\theta_{2}) & (1- \varepsilon) p(\theta_{3}) \\ \varepsilon & 1 - \varepsilon & 0 & 0 \\ \varepsilon &0 & 1- \varepsilon & 0 \\ \varepsilon & 0 & 0 & 1- \varepsilon
\end{pmatrix}.
\end{equation*}

\noindent To understand why the hypothetical belief matrix takes this form, note that if someone were to observe no information, then she would still understand this event to be an $\varepsilon$ probability occurrence; this corresponds to the entry in the upper left corner. On the other hand, since she obtains no information following the uninformative signal, the probability she would then assign to observing each of the other three signals, \emph{conditional on observing an informative signal}, simply coincides with the prior distribution. Since the probability of observing such a signal in the first place is $1- \varepsilon$, she therefore assigns $(1- \varepsilon)p(\theta_{i})$ to the event that she would have seen the signal revealing state $\theta_{i}$.  Following every other signal, while she \emph{would} know the state, she would also understand that there was an $\varepsilon$ chance that she would have remained uninformed. Thus, on the one hand, she assigns probability $\varepsilon$ to observing the uninformative signal, but also assigns probability 0 to observing any signal that would reveal any other state. Importantly, this matrix is row stochastic, each row itself being a probability distribution over $S$.
\end{example}

\subsection{Discussion} 
The framework above assumes a continuum population. This is for expositional simplicity. In general, as long as the $\abs{S}$ signals are observed, one can immediately write the matrix $B$. However, the matrix $Q$ is slightly harder to conceptualize; with a finite population, it represents the probability distribution over signals that the individual would expect given a conditionally independent draw from the same experiment. Equivalently, it represents the proportion of each belief type a decisionmaker would expect to observe given an infinite population. 

More substantive is that the matrix $B$ and $Q$ are observed without noise. If the analyst had access to an infinite sample, this could indeed be reasonable, but in practice one would expect $B$ and $Q$ to be observed only imprecisely. On the one hand, in the method outlined below, a small perturbation in $B$ and $Q$ will correspondingly lead to only a small change in the solved $\mathcal{I}$; thus, a slight error in $B$ or $Q$ will translate to a slight error in $\mathcal{I}$ (although how slight will itself depend on the accuracy $B$ and $Q$). 

Note that I also assume that the matrix $B$ does not have any column equal to 0. This assumption is not substantive. If posterior beliefs always put probability 0 on some state, then given updating from a common prior and a common Blackwell experiment, the prior probability on this state must also be 0. As my main case of interest is in when this is how beliefs are generated, the case where this assumption is violated is not relevant for my main results. The full rank assumption is more restrictive, and I comment on it in more detail in Appendix \ref{app:lineardependence}. There, I argue that the assumption that $B$ is full rank is generic, in a sense I describe, and show how the regression procedure I propose can be extended in cases where columns of $B$ feature linearly dependencies.

\section{Identifying Information if $\abs{S} \geq \abs{\Theta}$} \label{sect:mainresults}

The previous section showed how to construct both kinds of belief matrices from an information structure given a prior belief; the state belief matrix $B$ is computed from Bayes rule, whereas the hypothetical belief matrix additionally can be derived using rules of conditional and total probabilities. The goal of this paper is to go in the opposite direction. A well-known result (originally due to \cite{AM1995}) states that, given a prior belief and a set of posteriors, there is an essentially unique information structure inducing these beliefs. The main question I study in this paper is how to infer the decision-maker's information structure and prior using data from the (state and hypothetical) belief matrices. 

The following result provides an answer to this question when $\abs{S} \geq \abs{\Theta}$: 

\begin{theorem} \label{thm:identification} 
Suppose $B$ and $Q$ are generated by an informational environment $(\mathcal{I},p)$ satisfying the conditions in Section \ref{sect:model}, with $\abs{\Theta} \leq \abs{S}$. Then $\mathcal{I}$ is uniquely identified by $B$ and $Q$. If, for this $\mathcal{I}$, $\mathcal{I}B$ is irreducible, then $p$ is uniquely identified by $B$ and $Q$. Furthermore, the irreducibility condition holds for generic Blackwell experiments.\footnote{Here and in the proof, I take generic to mean that the set of Blackwell experiments which fail this condition fall within a lower dimensional subspace.}
\end{theorem}

The Theorem shows how to construct the information structure from the state belief matrix and the hypothetical belief matrix. Specifically, it shows that the information structure arises from \emph{regressing a given column of the hypothetical belief matrix on the columns of the state belief matrix.} I illustrate this using the truth-or-noise information structure from the previous section. If $p(\theta_{1})=p(\theta_{2})=p(\theta_{3})=1/3$, then I compute: 

\begin{equation*}
(B^{T}B)^{-1} B^{T} = \begin{pmatrix} \frac{1}{4} & \frac{11}{12} & - \frac{1}{12} & - \frac{1}{12} \\  \frac{1}{4} & - \frac{1}{12} &  \frac{11}{12} & - \frac{1}{12} \\ \frac{1}{4} & - \frac{1}{12} & - \frac{1}{12} &  \frac{11}{12}  \end{pmatrix}
\end{equation*}

\noindent For an arbitrary vector $v$, the expression $(B^{T}B)^{-1}B^{T}v$ is well-known---this is simplly the ordinary least squares regression of $v$ on $B$. That is, it  determines the coefficients $\beta$ in the equation $v=B \cdot \beta$ which provide the ``best-fit'' (according to the least square error). Also observe that in this example, the rows of $(B^{T}B)^{-1}B^{T}$ sum to 1 (which turns out to be a general property). Letting $\mathbf{1}_{n}$ denote a vector of 1s of length $n$, since the first column of $Q$ is $\varepsilon \cdot \mathbf{1}_{4}^{T}$, the coefficients corresponding to this regression are $\varepsilon \cdot \mathbf{1}_{3}^{T}$. This is exactly the vector of probabilities of observing the null signal in each of the three states. On the other hand, if one were to instead consider the second column of $Q$:

\begin{equation*} 
(B^{T}B)^{-1}B^{T} \begin{pmatrix} \frac{1- \varepsilon}{3}   \\  1- \varepsilon \\ 0  \\ 0 \end{pmatrix} = \begin{pmatrix} 1- \varepsilon \\ 0 \\  0  \end{pmatrix},  
\end{equation*}

\noindent which is precisely the vector of probabilities that the decision-maker observes the signal saying the state is $\theta_{1}$ (that is, the probability that this signal is seen in each of the three states). 

Notice that Theorem \ref{thm:identification}, and in particular the regression interpretation, also clarifies \emph{exactly} what is decided by each column of $Q$. Each column of $Q$ gives a unique column of the matrix determining the information structure $\mathcal{I}$---in particular, the column it yields is the vector of probabilities (with different coordinates corresponding to different states) that that signal emerges. This property will be important when I compare my approach to an alternative one in Section \ref{sect:SP}. As Section \ref{sect:generic} notes, however, not \emph{every} (plausible) $Q$ will yield an information structure consistent with $B$.

Having described how to obtain the information structure, I now consider the prior. An information structure can be identified by the prior belief and the set of posteriors induced by it (which in this case is $B$). However, this requires the prior to be in the interior of the convex hull of the posterior beliefs. While the regression procedure finds the information structure, it does not guarantee that the resulting prior satisfies the interiority condition. 

Addressing this issue yields another insight into the geometry of information structures. Note that the probability of observing any particular signal determined from $\mathcal{I}$ and $p$: 

\begin{equation*} 
\sum_{\theta} \mathcal{I}(\theta)[s] \cdot p(\theta) = \Prob[s]. 
\end{equation*} 

\noindent The martingale property of beliefs states that: 

\begin{equation*} 
\sum_{s} b_{\theta,s} \Prob[s] = p(\theta).
\end{equation*} 

\noindent Thus, substituting in for $\Prob[s]$ and rewriting in matrix form gives the following identity: 

\begin{equation*} 
B^{T}\mathcal{I}^{T}p=p \Rightarrow (B^{T} \mathcal{I}^{T} - I)p = p.
\end{equation*}

\noindent This equation demonstrates that the prior is therefore a \emph{unit eigenvector with eigenvalue 1 of the matrix} $B^{T}\mathcal{I}^{T}$ (or, by taking transposes, a left eigenvector of $\mathcal{I}B$).\footnote{By unit eigenvector, in this paper I mean the entries sum to 1; such an eigenvector has a norm equal to 1 when $\abs{\abs{x}}=\sum_{i} \abs{x_{i}}$.}  And in fact, given the previous observation that the information structure $\mathcal{I}$ can be identified from $B$ and $Q$, this shows that the prior can as well. The proof verifies that indeed this eigenvalue can be guaranteed to exist, via an appeal to the Frobenius-Perron theorem. It is worth emphasizing that the Frobenius-Perron theorem ensures that the eigenvector which satisfies this equation is \emph{unique} (up to scaling), implying that the prior is pinned down as well.

This derivation introduces the object $B^{T} \mathcal{I}^{T}$, whose $i$th row and $j$th column is: 

\begin{equation*} 
\sum_{\tilde{s}}  \mathcal{I}(\theta_{j})[\tilde{s}]b_{\theta_{i},\tilde{s}}. 
\end{equation*} 

\noindent The represents the expectation, in state $\theta_{j}$, of the probability assigned to state $\theta_{i}$ by a population member selected uniformly at random. Roughly speaking, with a more informative experiment, this should be larger for $\theta_{i}= \theta_{j}$ and smaller for $\theta_{i} \neq \theta_{j}$. By comparing the diagonal of the matrix $\mathcal{B}^{T} \mathcal{I}^{T}$ to the other entries, the analyst can get a rough sense of the accuracy of the beliefs in the population. 

A dual operation would be to instead substitute in for $p(\theta)$ instead of $\mathbb{P}[s]$; this operation is dual in the sense that the matrix obtained in the analogous eigenvector equation is $B \cdot \mathcal{I}$ in place of $B^{T} \cdot \mathcal{I}^{T}$. Inspecting this equation, however, note that it is precisely the hypothetical belief matrix $Q$. This result, that $\mathbb{P}[s]$ is the unit eigenvector of $Q$ with eigenvalue 1, appeared in \cite{PSM2013}, and further exploited by \cite{PM2022}; notice that the derivation of $\mathbb{P}[s]$ is thus more direct than the derivation of $p(\theta)$, as $Q$ is a primitive of the environment. I return to a discussion of this contrast in Section \ref{sect:SP}.

The condition needed for uniqueness of the prior is that the matrix of which the prior is a Perron eigenvector is irreducible. See the discussion in Appendix \ref{app:irreducibility} for a discussion of what is identified when irreducibility fails. Briefly, if $B^{T}\mathcal{I}^{T}$ is not irreducible, then one can find a Perron eigenvector for each irreducible class of the stochastic process which $B^{T}\mathcal{I}^{T}$ defines. In this case, any convex combination between these eigenvectors will be a prior inducing $B$ and $Q$.\footnote{Note that if $B^{T} \mathcal{I}^{T}$ is not irreducible, then $Q$ will not be either; this follows from the observation that the signal space can be partitioned using the states that can be distinguished. Similarly, the only way a decisionmaker can know a certain signal would never be observed in the population is if they also know that any state where that signal is drawn is not the true state; so, if $Q$ is not irreducible, then neither is $B^{T}\mathcal{I}^{T}$.} As implied by the above discussion, irreducibility will fail if, for instance, an agent can always distinguish between two subsets of the state space.\footnote{To clarify this most transparently, Appendix \ref{app:irreducibility} focuses on the case of deterministic information structures, where signals are generated deterministically as a function of the state. In this case, no information is conveyed within each irreducible class---i.e., within each partition element---even though it is evident that agents can distinguish between elements of the partition.} In this case, the prior can only be uniquely defined within each irreducible class, but not overall.

\subsection{Are Generic $B$ and $Q$ Generated by Informational Environments?} \label{sect:generic}

A natural question that emerges from the previous discussion is how restricted the set of hypothetical belief matrices are. Can an arbitrary plausible hypothetical belief matrix be consistent with a state belief matrix? The answer to this question turns out to be no. For instance, suppose the belief matrix induced is the following: 

\begin{equation*} 
B= \begin{pmatrix} 1/4 & 3/4 \\ 3/4 & 1/4 \end{pmatrix} \Rightarrow (B^{T}B)^{-1} B^{T}= \begin{pmatrix} - \frac{1}{2} & \frac{3}{2} \\ \frac{3}{2} & - \frac{1}{2} \end{pmatrix}.
\end{equation*}

\noindent Now, given some candidate $Q= \begin{pmatrix} a & 1-a \\ 1-b & b \end{pmatrix}$, right multiplying by $Q$ gives: 

\begin{equation*} 
\begin{pmatrix} - \frac{a}{2}+ \frac{3(1-b)}{2} &  - \frac{1-a}{2}+ \frac{3b}{2} \\  \frac{3a}{2}-\frac{1-b}{2} &   \frac{3(1-a)}{2}- \frac{b}{2} \end{pmatrix}.
\end{equation*}

\noindent Now, $Q$ is plausible whenever $a, b \in [1/4,3/4]$; yet a corollary of Proposition \ref{prop:dimproof} is that, given a value for $a$, only a single value for $b$ will be consistent with a belief landscape. Thus, while any such value for $a$ and $b$ would yield an information structure, it does not necessarily yield one that can induce $B$ given any prior $p$. The prior, together with $B$, pins down the information structure; from this, the matrix $Q$ can always be inferred. These observations yield the following Proposition:

\begin{prop}  \label{prop:dimproof} 
Consider a belief matrix $B$ with linearly independent columns. The set of all possible $Q$ such that $B$ and $Q$ are generated by an informational environment is $\abs{\Theta}-1$ dimensional.  
\end{prop}

Thus for binary state informational environments, the set of $Q$ inducing $B$ is one dimensional, verifying the claim that given a possibly valid $a$, there is a unique value of $b$ which corresponds to a valid information structure. Figure \ref{fig:infoex} shows how, given the belief matrix $B$, which $a$ and $b$ choices correspond to a fixed feasible prior $p$. For instance, $a=b=5/8$ is the solution when $p=1/2$. The choice of $a=b=9/16$ is therefore invalid; nevertheless, for these choices: 

\begin{equation*} 
(B^{T}B)^{-1}B^{T}Q=\begin{pmatrix} 3/8 & 5/8 \\ 5/8 & 3/8 \end{pmatrix}
\end{equation*}

\noindent Upon inspection, one can see that this \emph{is} indeed a perfectly valid information structure $\mathcal{I}$, and in fact one that is symmetric. But, it is also straightforward to see that it cannot induce the belief matrix $B$ for any prior; indeed, since $B$ is symmetric as well, symmetry would require $p=1/2$, which in turn would suggest distinct beliefs given the information structure than those given by $B$. Thus, while the choices of $a=b=9/16$ induce a $Q$ matrix which is row stochastic and non-negative, and is such that $(B^{T}B)^{-1}B^{T}Q$ is an information structure, it cannot be generated by any informational environment.

\begin{figure}
\centering
        \includegraphics[height=0.23\textheight]{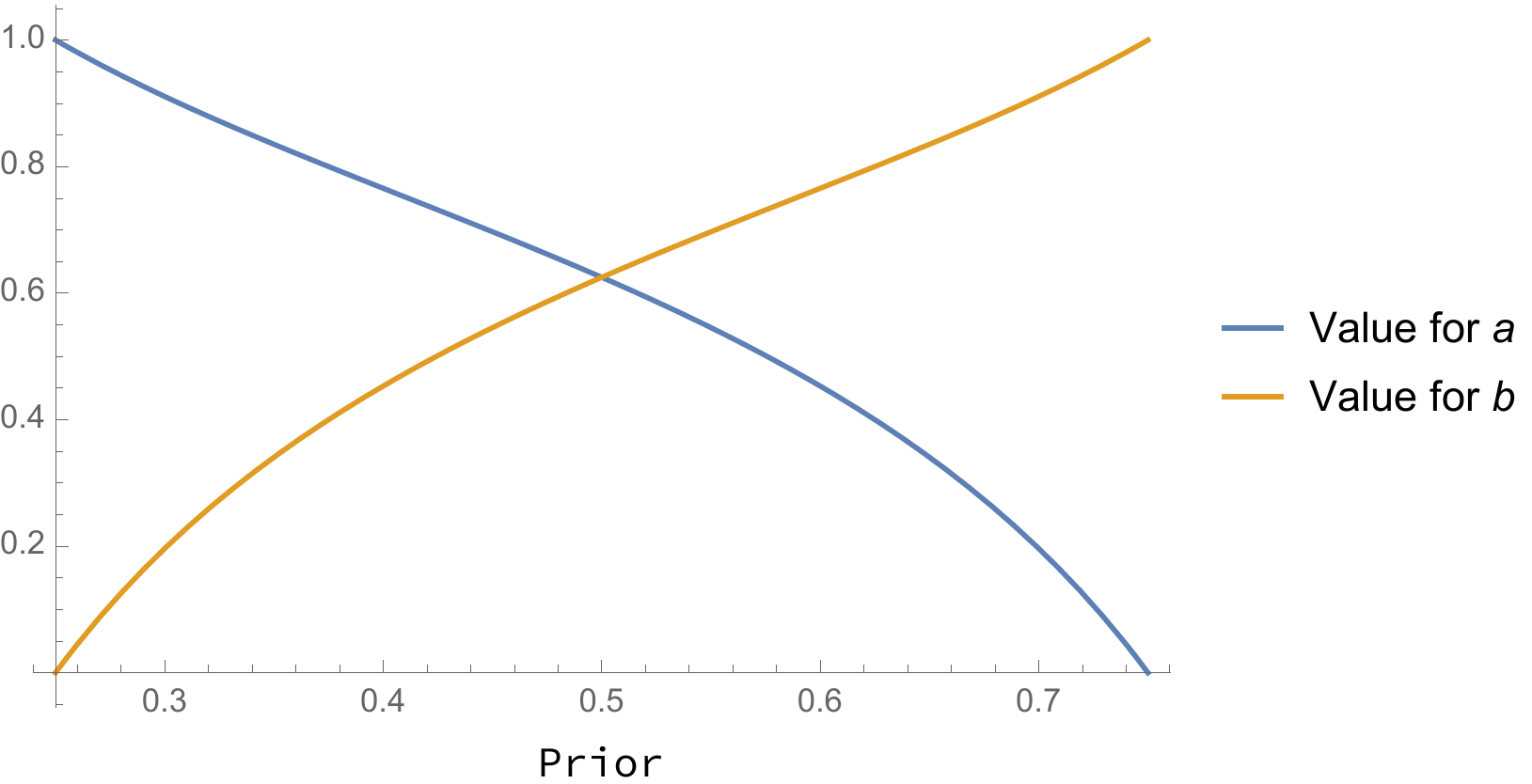}
         \caption{Value for $a$ and $b$ in the example which yield a valid hypothetical belief matrix, given a prior probability of state 1 equal to $p \in[1/4,3/4]$.}
        \label{fig:infoex}
\end{figure}

One observation that is helpful for interpreting this statement is the following: 

\begin{prop} \label{prop:dimQ}
Given a row stochastic $B$, the space of plausible $Q$ (as defined in Section \ref{sect:model}) such that the regression procedure in Theorem \ref{thm:identification} induces an information structure has dimension $\abs{S}(\abs{S}-1)$. 
\end{prop} 

\noindent The proof shows that if $Q$ is a matrix rows sum to $1$, then $(B^{T}B)^{-1}B^{T}Q$ will be a matrix whose rows sum to 1. The result then follows by noting that one can find an open set of plausible $Q$ inducing some valid information structure within the set of row stochastic $Q$ under the relative topology.\footnote{That is, viewing a row stochastic matrix of dimension $n$ as a subset of $\mathbb{R}^{n^{2}}$ and considering the relative topology on the set of row-stochastic matrices.} Note that the conditions of the theorem restrict consideration to $Q$ such that $(B^{T}B)^{-1}B^{T}Q$ has exclusively non-negative entries; this need not hold for arbitrary plausible $Q$,\footnote{Consider, for instance, Example \ref{ex:mainillustration} with $p(\theta_{1})= p(\theta_{2})=p(\theta_{3})=1/3$ and $\varepsilon=1/2$. The matrix $\tilde{Q}= \begin{pmatrix} 1/2 & 1/6 & 1/6 & 1/6 \\ 1/2 & 1/2 & 0 & 0 \\ 1/2 & 0 & 1/2 & 0 \\ 1/4 & 1/4 & 1/4 & 1/4 \end{pmatrix}$ is plausible, but $(B^{T}B)^{-1} B^{T}\tilde{Q}= \begin{pmatrix} 25/48 & 23/48 & -1/48 & 1/48 \\ 25/48 & -1/48 & 23/48 & 1/48 \\ 13/48 & 11/48 & 11/48 & 13/48 \end{pmatrix}$.} meaning that not all plausible $Q$ necessarily yield information structures. Still, the set of $Q$ which do is of the same dimensionality. 

An implication of the previous two results is that the space of $B,Q$ generated by informational environments is non-generic in the space of plausible $B$ and $Q$. Intuitively, this restriction is due to the prior.  A matrix $B$ and prior $p$ determine an information structure, so the set of possible information structures inducing $B$ has the same dimension as the set of possible priors. But the set of possible $Q$ that could emerge is \emph{much} larger than the set of possible priors. The next section further underscores this point.

\subsection{Relaxing Common Priors} 

Note that an assumption of this framework is that beliefs are updated from a fixed prior, $p$; this property, in turn, is part of what drives the dimensionality of the set of $Q$, since $\abs{\Theta}-1$ is the set of priors which induce a fixed belief matrix $B$ (which will have linearly independent columns whenever $B$ is full rank). This raises the question of how restrictive this assumption is; I now comment how relaxations of the common prior can rationalize any $Q$. 

I maintain the framework that there are $\abs{S}$ belief types in the population, with state belief matrix $B$ and hypothetical belief matrix $Q$. I emphasize that the $k$th row of the state belief matrix represents the $k$th belief type's probability assessment over $\Theta$; the $k$th row of the hypothetical belief matrix, as before, represents the $k$th belief type's expected proportion of belief types in the population. 

However, rather than suppose that each belief type is generated by a fixed information structure $\mathcal{I}$ and prior $p$, consider the case where each belief type $s$ maintains that \emph{their} beliefs were arrived at from updating a belief-type dependent prior, say $p^{s} \in \Delta(\Theta)$, and belief-type-dependent information structure, say $\mathcal{I}^{s} : \Theta \rightarrow \Delta(S)$. Under these assumptions, the matrix $B$ and $Q$ could always be rationalized as follows: 

\begin{equation} 
p^{s} = (b_{s, \theta})_{\theta \in \Theta}, ~~~ \mathcal{I}^{s}(\theta)[\tilde{s}]=q_{s,\tilde{s}}; \label{eq:anything}
\end{equation}

\noindent that is, where every individual in the population thinks signals are uninformative, yet different belief types arrive at distinct posterior beliefs due to their differing priors.  Note that this alternative not only allows for disagreement over priors, but also disagreement over the informational content. 

A seemingly more constrained formulation is to assume that while the entire population updates beliefs using the same information structure, each belief type also has its own prior. That is, consider the same framework above, but assume that $\mathcal{I}$ does not depend on the belief type---thus, the population agrees over the informational content of $\mathcal{I}$, even though belief type $s$ updates using prior $p^{s}$. Despite appearances, it turns out that this alternative is no less constrained than one where individuals can disagree over the informational content of signals. The following result shows that this framework can completely rationalize any $Q$: 

\begin{prop} \label{prop:noncommon}
Suppose $B$ and $Q$ are plausible. Then there exists priors for each belief type, $p^{s} \in \Delta(\Theta)$, such that the state belief matrix $B$ and hypothetical belief matrix $Q$ emerge via Bayesian updating using prior $p^{s}$ and some information structure $\mathcal{I}$.  
\end{prop}

While the observation in this section are suggestive that allowing for non-common priors can rationalize arbitrary $B$ and $Q$, ultimately my view is that the framework present is not satisfactory and that it may be necessary to elicit richer objects. For instance: While the information environment in (\ref{eq:anything}) yields identical $B$ and $Q$ as the information environment in Proposition \ref{prop:noncommon}, one can distinguish between them by asking someone of belief type $s$, ``what probability would you have assigned to state $\theta$ had your signal been $s' \neq s$?'' If information were generated as in  (\ref{eq:anything}), someone observing a signal $s$ would answer this question with $b_{s,\theta}$, for all $s'$, since in that example the entire population views signals as drawn independently from $\theta$. However, if $\mathcal{I} = (B^{T}B)^{-1}B^{T}Q$ is informative, then beliefs should respond to signals and not necessarily be constant, provided the prior is non-degenerate. While these kinds of questions may or may not be natural depending on the application, they are reminiscent of designs in experimental work employing the strategy method (see \cite{BrandtsCharness2011}), whereby a subject would be asked to state their beliefs following any possible signal realization that might be potentially observed (whether or not it actually is).

Proposition \ref{prop:noncommon} clarifies one way of interpreting $(B^{T}B)^{-1}B^{T}Q$ in cases where $Q$ lies outside of the $\abs{\Theta}-1$ dimensional subspace identified in Proposition \ref{prop:dimproof}. In these cases, it need not be possible to find a single prior such that, when individuals adopt $\mathcal{I}$ and update from this prior, $B$ and $Q$ are induced. On the other hand, this is possible if one allows for each belief type to have its own prior. 

I can illustrate this idea using the example from Section \ref{sect:generic}. Recall that when $a=b=9/16$, $(B^{T}B)^{-1}B^{T}Q$ cannot induce $B$ and $Q$ given a single prior. However, suppose signal $s_{1}$ (corresponding to the first row of $B$) starts with prior $(5/14,9/14)$. In this case, the probability assigned to the first state, when updating according to $(B^{T}B)^{-1}B^{T}Q$ is $\frac{(5/14)(3/8)}{(5/14)(3/8)+(9/14)(5/8)}=1/4$, which coincides with the relevant entry from $B$. And indeed, $(1/4)(3/8)+(3/4)(5/8)=9/16$, coinciding with relevant entry from $Q$. Thus, with this prior, belief type $s_{1}$ holds the appropriate state beliefs and hypothetical beliefs, provided the information structure is $(B^{T}B)^{-1}B^{T}Q$. Identical calculations show that, if signal $s_{2}$ (corresponding to the second row of $B$) were to start with prior $(9/14,5/14)$, then state beliefs would be the second row of $B$ and hypothetical beliefs would be the second row of $Q$. Thus, while $B$ and $Q$ could not be induced given a common prior, they can be induced via a model with non-common priors. Under these assumptions, the regression procedure will still identify $\mathcal{I}$.

\section{The $|\Theta| > |S|$ case}  \label{sect:morestates}

The regression characterizations no longer apply when there are more states than signals, in the same way that an Ordinary Least Squares requires more observations than covariates in order to yield an identified solution (i.e., $B$ must have more rows than columns, in addition to having no linear dependencies). The issue with my approach is that the question used to identify $\mathcal{I}$ has multiple solutions. Letting 
 $v$ be any vector in the null space of $B$, for any candidate solution $\mathcal{I}(\cdot)[\tilde{s}]$: 

\begin{equation} 
q_{\cdot,\tilde{s}}=B \cdot\mathcal{I}(\cdot)[\tilde{s}] \Rightarrow q_{\cdot,\tilde{s}}=B\cdot(\mathcal{I}(\cdot)[\tilde{s}]+v).  \label{eq:equationset}
\end{equation}

\noindent Thus, the set of solutions to the equation used to identify the column of $\mathcal{I}$ corresponding to $\tilde{s}$ has dimension equal to the null space of $B$, which is (generically) $\abs{\Theta}-\abs{S}$ dimensional.

Despite this multiplicity, one can still use regression ideas in order to determine the subspace of matrices which satisfy (\ref{eq:equationset}), even though $B^{T}B$ is not invertible when $\abs{\Theta}> \abs{S}$. Seeking to solve for a candidate information structure, as before, pre-multiply the matrix equation for the hypothetical beliefs by the state belief matrix $B$: 

\begin{equation*} 
B^{T}Q=B^{T}B \mathcal{I}
\end{equation*}

\noindent If $\abs{S}< \abs{\Theta}$, $B^{T}B$ is not invertible. The idea is to augment this equation so that an inversion can occur. Adding $M\mathcal{I}$, for any matrix $I$, to both sides of this equation:

\begin{equation*} 
B^{T}Q+M\mathcal{I}= (B^{T}B+M) \mathcal{I}.
\end{equation*}

\noindent If $B^{T}B + M$ is invertible, one can obtain an expression where $\mathcal{I}$ only appears on the right hand side (multiplied by a matrix factor). A natural choice is $M = \lambda I$, where $I$ is the identity matrix, yielding:

\begin{equation} 
B^{T}Q + \lambda \mathcal{I} = (B^{T}B + \lambda I) \mathcal{I} \Rightarrow (B^{T}B+ \lambda I)^{-1}B^{T}Q = \mathcal{I}(1- (B^{T}B + \lambda I)^{-1} \lambda I) \label{eq:ridgeequation}
\end{equation}

\noindent  A necessary and sufficient condition for the right hand side to converge to $\mathcal{I}$ as $\lambda \rightarrow 0$ is:

\begin{equation} 
(B^{T}B + \lambda I)^{-1} \lambda I \rightarrow 0. \label{eq:limit}
\end{equation}

\noindent  In statistics, the idea of adding $\lambda I$ in order to be able to invert $B^{T}B$ is used to motivate \emph{ridge regression}; in cases where more general perturbations are considered, this process is also known as \emph{Tikhonov regularization.}

Unfortunately, the condition (\ref{eq:limit}) is not necessarily satisfied; but as Theorem \ref{thm:morestates} shows, the solution in the limit as $\lambda \rightarrow 0$ will nevertheless be a solution to $Q=B\mathcal{I}$.\footnote{In statistical applications where this method is used, taking $\lambda$ too small is often undesirable. See \cite{RidgeTextbook} for more on what guides the choice of $\lambda$ in practice.} 

\begin{theorem} \label{thm:morestates}
Suppose $\abs{\Theta} > \abs{S}$. The matrix: 
\begin{equation*}
\tilde{\mathcal{I}}=\lim_{\lambda \rightarrow 0} (B^{T}B + \lambda I)^{-1}B^{T}Q,
\end{equation*} 

\noindent exists, is well-defined, and solves the equation $B=Q\tilde{\mathcal{I}}$. Furthermore, 

\begin{itemize} 
\item Let $v_{1}, v_{2}, \ldots, v_{k}$ be a basis for the null-space of $B$. Then $\mathcal{I}(\cdot)[s_{i}]=\tilde{\mathcal{I}}(\cdot)[s_{i}]+\sum_{j} \alpha_{j} v_{j}$, for some $\alpha_{1},\ldots, \alpha_{k} \in \mathbb{R}$. 

\item The prior $p$ is the unique (unit) eigenvector of $B^{T}\tilde{\mathcal{I}}^{T}$ with eigenvalue 1.

\end{itemize} 

\end{theorem}

\noindent The Theorem shows that the $\lambda \rightarrow 0$ limit is equal to a solution of the equation defining hypothetical beliefs in terms of state beliefs; it turns out that this property is sufficient in order to recover the eigenvector interpretation of the prior, \emph{whether or not the solution is in fact the true information structure}. While this does not imply that $\tilde{\mathcal{I}}$ is itself the information structure inducing $B$---or, for that matter, even an information structure---as per the discussion above, this does nevertheless provide some meaningful information which can in some cases be used to pin down the information structure. 

The proof shows that the same argument illustrating that the prior is a unit eigenvector of $B^{T}\mathcal{I}^{T}$ with eigenvalue 1 can be used to show that it is an eigenvector of $B^{T} \tilde{\mathcal{I}}^{T}$, and proceeds to show that indeed one must exist. The following example illustrates: 

\begin{example} 
Suppose: 

\begin{equation*} 
B= \begin{pmatrix} \frac{2}{3} & \frac{1}{3} & 0 \\ \frac{4}{9} & \frac{1}{9} & \frac{4}{9} \end{pmatrix}, ~~~~~~ Q= \begin{pmatrix} \frac{7}{18} & \frac{11}{18} \\ \frac{11}{54} & \frac{43}{54} \end{pmatrix}. 
\end{equation*}

\noindent In this case, I compute $\tilde{\mathcal{I}}$ to be:

\begin{equation*} 
\tilde{\mathcal{I}}=\begin{pmatrix} \frac{29}{63} & \frac{52}{63} \\ \frac{31}{126} & \frac{23}{126} \\ - \frac{4}{63} & \frac{58}{63} \end{pmatrix}
\end{equation*}

Obviously, $\tilde{\mathcal{I}}$ is not an information structure, as each row violates the two requirements to be probability distributions: non-negative entries and summing to 1. The nullspace of $B$ is spanned by a single vector, $(-2,4,1)$. Adding a multiple of this vector to the first column of $\tilde{\mathcal{I}}$ to make sure all entries are non-negative yields:

\begin{equation*} 
\begin{pmatrix} \frac{1}{3} & \frac{52}{63} \\ \frac{1}{2} & \frac{23}{126} \\ 0 & \frac{58}{63} \end{pmatrix}
\end{equation*}

\noindent Adding another multiple of $(-2,4,1)$ to the second column allows each row to sum to 1, and yields: 

\begin{equation*} 
\mathcal{I}= \begin{pmatrix} \frac{1}{3} & \frac{2}{3} \\ \frac{1}{2} & \frac{1}{2} \\ 0 & 1 \end{pmatrix}
\end{equation*}

\noindent which indeed induces these belief matrices. I compute: 

\begin{equation*} 
B^{T}\mathcal{I}^{T}= \begin{pmatrix} 14/27 & 5/9 & 4/9 \\ 5/27 & 2/9 & 1/9 \\ 8/27 & 2/9 & 4/9 \end{pmatrix} ,~~~~ B^{T} \tilde{\mathcal{I}}^{T} = 
 \frac{1}{567}\begin{pmatrix} 382 & 139 & 208 \\ 139 & 58 & 46 \\ 208 & 46 & 232
\end{pmatrix}.
\end{equation*}

\noindent Indeed, for both of these matrices, there is a unique (up to scaling) eigenvector with eigenvalue 1, and it is $(1/2,1/6,1/3)$. Bayes rule verifies that this matrix, together with $\mathcal{I}$, induces $B$. 
\end{example}

\medskip

Note that knowledge of the prior would allow the analyst to learn the information structure itself. That is, given the prior $p(\theta)$, the analyst could determine signal weights $(\alpha_{s})_{s \in S}$ such that $p(\theta)=\sum_{s} \alpha_{s} b_{\theta,s}$ (which exist since $p(\theta)$, as the prior, is in the interior of the convex hull of the columns of $B$). As pointed out in Proposition 1 of \cite{KG2011}, one then has $\mathcal{I}(\theta)[s]=b_{\theta,s}\alpha_{s}/p(\theta)$.

Despite this observation, as I make clear in the next Section, I view Theorem \ref{thm:morestates} as largely of theoretical interest, as it demonstrates that the martingale condition for beliefs pins down the prior given any candidate solution to the equation defining hypothetical beliefs. While this helps clarify what drives the identification of the prior, I do not wish to overstate its practical significance; it seems more direct to simply obtains $\mathbb{P}[s]$ from $Q$ rather than $p(\theta)$ from $B^{T}\mathcal{I}^{T}$; as I discuss in the next section, this provides a more immediate path to the information structure. The key point is that (i) regression techniques can still recover solutions for $I$ of the equation $Q=BI$, even though (ii) with more states than signals, that may not be enough to recover the information structure, but that (iii) the eigenvector interpretation of the prior remains even in this case. 

\section{Comparison to The Signal Priors Approach}  \label{sect:SP}
\cite{PSM2013,PSM2017} and \cite{PM2022} provide an alternate algorithm which can be used to derive the matrix $\mathcal{I}$.\footnote{I am grateful to a referee who provided essentially all of the insights of this section.} As alluded to above, the crux of this approach is to use $Q$ directly to obtain the (ex-ante) distribution over the signals in the population. That is: Given the matrix $Q$, note that if $v$ is an eigenvector of $Q$ with eigenvalue 1 (which, again, exists by Frobenius-Perron), then $v$ is proportional to $\Prob[s]$; indeed, $\sum_{s} \Prob[s^{*} \mid s]\Prob[s]= \Prob[s^{*}]$. Thus, normalizing the eigenvector of $Q$ so that the entries sum to 1 provides a way of arriving at the (ex-ante) probability distribution over \emph{belief types}. In this case, one has $p(\theta)\mathcal{I}(\theta)[s]=b_{s,\theta}\Prob[s]$. To obtain $p(\theta)$, simply consider $\sum_{\tilde{s}}b_{\tilde{s},\theta}\Prob[\tilde{s}]$. This delivers the information structure $\mathcal{I}$.

This method, called the Signal Priors Approach (henceforth SP) is complementary to the regression approach, and each may have merits depending on the application. For instance, while SP clearly requires $Q$, it can identify the relevant signal distribution even without any information about $B$ at all. Hence if $Q$ is fully specified and $B$ is not, then this approach allows for part of the information structure to be recovered nevertheless. 

On the other hand, there are instances where the regression approach of this paper may be beneficial. Perhaps most significantly, it may be that $B$ is fully specified, but not all of $Q$ is. In that case, one can compute $(B^{T}B)^{-1}B^{T}q$ given a column $q$, and still find the probability that that signal is induced in each state $\theta$. Note that, while this would not identify the full information structure, this procedure could still be used to determine the state $\theta$, the typical focus of the crowd wisdom problem. More precisely: if $q$ corresponds to signal $s$, and if $(B^{T}B)^{-1}B^{T}q$ yields $\abs{\Theta}$ distinct entries, then the analyst could learn the true $\theta$ by seeing which entry matched the empirically observed proportion of the population receiving signal $s$. 

While the extent to which this would work in practice is difficult to assess theoretically (e.g., it may depend on noise in the population, or how many beliefs are elicited), one can imagine there being cases where eliciting the full matrix $Q$ is significantly more difficult to obtain than a single column of it. For instance, suppose there are a very large number of possible signals, and an analyst needs to elicit the probabilities of each signal type from the population. In this case, it  might be prohibitive to elicit, from each individual, the \emph{entire} vector of proportions that each signal appears in the population. However, rather than eliciting $\abs{\Theta} + \abs{S}$ beliefs, as SP requires, my method implies that one could (generically) learn the state by eliciting only $\abs{\Theta}+1$ beliefs. 

This discussion suggests that the regression approach is perhaps more suitable when the number of signals is larger than the number of states. However, the results I provide on the case of $\abs{\Theta} > \abs{S}$ are largely theoretical, in that they suggest properties of the information structure, but actual computation of it can be more cumbersome; more to the point, if the analyst were able to elicit beliefs over $\Theta$, then the $\abs{\Theta} > \abs{S}$ case is precisely when eliciting beliefs over $S$ entails fewer questions than eliciting beliefs over $\Theta$ in the first place. Thus, it appears that SP is more appropriate from a practical perspective when $\abs{\Theta} > \abs{S}$.



\section{Related Literature} 
\subsection{Wisdom of the Crowd and Predicting Peer Predictions} 
As discussed above, the paper borrows the framework of \cite{PSM2017} and \cite{PM2022}, which studied how an analyst may be able to infer a true state using predictions regarding the reports of other members of a population. Their use of higher-order beliefs built on \cite{Prelec2004}, which proposed using this additional data in order to induce \emph{truthful reporting} (as opposed to aggregating information). \cite{PSM2017} showed how to infer the true state using a mechanism referred to as the \emph{surprisingly popular algorithm}, and showed how it could aggregate information in the two-state case, as well as some other special cases. The generalization to arbitrary numbers of signals and states, discussed in Section \ref{sect:SP}, was presented in \cite{PM2022}, contemporaneous with this paper. Results illustrating the general difficulty of learning the true state \emph{without} such information also appear in  \cite{ABS2020}.

Among the papers in this literature, most related to my approach is the contemporaneous paper of \cite{CMP2021}. Their framework takes as a primitive an infinite population and a probability distribution over states and signals, without presuming the structure of an informational environment as described in Definition \ref{def:informationenviron} (whereas this paper is focused on recovering this experiment when it exists). Their main results show how, in the large sample limit, inverting a matrix of average beliefs allows the analyst to asymptotically recover the true state using predictions of the average population beliefs. They do so using assumptions which roughly speaking allow law of large number arguments to be applied. This matrix inversion is reminiscent of regression; however, unlike this paper and \cite{PM2022}, their proposal with a general number of states essentially requires as many distinct beliefs as states to be observed. My assumptions on signals and states, by contrast, implies a simple inversion need not be feasible. 

I have not considered the practical elicitation of the matrices $B$ and $Q$, as was the focus of \cite{Prelec2004}. See, for instance, \cite{CPRT2019} or
\cite{WP2012b} which suggests proposals for how to elicit this information given natural limitations an analyst might face (especially finiteness of the population). Robustness with respect to common priors is also discussed in \cite{WP2012a}, although recall that in the context of my exercise, relaxing common priors too strongly essentially implies no restrictions at all on $B$ and $Q$.

\subsection{Identification of Information}

This paper joins a literature which studies the identification of information. Typically, the focus in this literature is on what can be identified by a given dataset consisting of choices. By contrast, my approach takes beliefs themselves as a primitive, and focuses on the linear relationships between the information structure and the (state and hypothetical)  beliefs in the population. An example of the former approach is \cite{AM2017}, who study how to identify a decision-maker's belief over states given a decision problem. Closer to the current paper is \cite{Lu2016}, who asks when it is possible to identify a decision-maker's information, \emph{together with their utility function} from \emph{choice} data, if their stochastic choice is generated by a distribution over posterior beliefs; this is possible in certain restrictive cases. The question of identifying the information structure from choices in games featuring a linear quadratic normal structure is considered by \cite{Miyashita2022}, building on a framework of \cite{BergemannMorris2013} which also considered a similar question. Similar  exercises aimed at identifying information are carried out in   
\cite{LomysTarantino2022} (studying search) and \cite{Jakobsen21} (studying Bayesian Persuasion). 

\subsection{Regression Methods}

Methodologically, this paper shows how ideas from regression (and linear algebra more generally) facilitate tractable analysis of otherwise complicated phenomena. The analysis of \cite{Miyashita2022} in linear quadratic normal games relies upon similar ideas, although the data available is more restricted (and in particular, often confounded by the game structure) and as such the exact information structure can often not be identified. \cite{Caradonna2021} shows how regression determines a ``best fitting'' preference representation within a class of preferences, using this to quantify the predictive accuracy of a model. The applicability of these methods are likely not limited to statistical inference exercises; for instance, \cite{Ball2021} considers a model where an intermediary assigns a score to a strategic sender interested in biasing a prediction that is subsequently made based on the score. In his setting, the intermediary's objective reduces to a regression equation. 

\subsection{Priors as Eigenvectors} 
Aside from the complementary results in \cite{PM2022} (as well as \cite{PSM2013})   that $(\mathbb{P}[s])_{s \in S}$ has an eigenvector interpretation, my observation that the prior (over $\Theta$) has an eigenvector interpretation is perhaps most reminiscent of results \cite{SametIteratedExp}, further explored by \cite{SametSeparation} and \cite{GolubMorris}. \cite{SametIteratedExp} characterized a \emph{common} prior as the eigenvector of a stochastic transition matrix obtained from an information structure involving multiple agents (specifically, Propositions 3 and 5).\footnote{Specifically, the existence of the common prior is equivalent to this eigenvector being independent of the order in which agents are considered in defining this stochastic transition matrix. See \cite{Hellman2011} for a generalization of this result to infinite spaces.} These results provide an \emph{interim} characterization of the common prior assumption, i.e., after signals have been observed.\footnote{Several other papers have used either the Markov chain interpretation of interim beliefs introduced by \cite{SametIteratedExp}, or properties implied by the stationary distribution characterization of the common prior; see, for instance, \cite{MorrisShin}, \cite{CEMS2008}, or \cite{AngeletosLian}.}  \cite{GolubMorris} explore this idea further, roughly speaking using similar ideas to move from studying ``beliefs about beliefs'' to ``expectations about expectations''---characterizing implications, for instance, of assuming these expectations do not depend on the order agents are considered. They invoke the Frobenius-Perron theorem to characterize the ``consensus expectation'' with multiple agents.  

This paper's motivation is similar, though I am not only interested in identifying the prior from interim beliefs but also the information structure (whereas \cite{SametIteratedExp} does not assume a common prior but is interested in \emph{characterizing} when one exists). I note that strictly speaking the relevant matrices also differ, in part since the formal framework in this paper differs from \cite{SametIteratedExp}; there, agents observe an element of a partition of the state space, which in particular is generated \emph{deterministically} as a function of the state. That said, 
 higher-order beliefs are less restricted since it is not assumed that beliefs are generated in a conditionally-IID manner. These differences are significant, to the point where translating one result directly to the other does not seem immediate.

\section{Conclusion}
It is typically not possible for an analyst to identify an information structure from the set of possible ex-post beliefs over the states which generate them---any prior in the interior of the convex hull of these beliefs could generate those signals under an appropriately chosen information structure. This paper shows not only that the information structure and prior can be recovered using data obtainable after beliefs are formed---by asking about the perceived proportion of beliefs in the population---but also uncovered some geometric logic underlying these objects. First, the information structure can be identified (at least when there are as many signals as states) using a regression procedure, and second, that the prior can be identified as the eigenvector of a matrix obtained using the identified information structure.  

One potentially disconcerting feature of the analysis is that the dimensionality of the set of plausible hypothetical belief matrices $Q$ is much larger relative to the dimensionality of the set of hypothetical belief matrices $Q$ which could emerge given $\mathcal{I}$. This suggests that the underlying framework may be easily falsified in practice---if individuals only report beliefs imprecisely, then it may be that the observed $Q$ falls outside of this smaller dimensional space. While there are ways this could be accommodated (e.g., by relaxing common priors), it may be desirable to have a relaxation which would accommodate richer $Q$ without simply stating that ``anything goes.'' These questions seem worthwhile but are beyond the scope of the current paper. 

Two further avenues seem worthwhile of pursuit. First, the comparison to the Signal Priors approach suggests that there are multiple methods at an analyst's disposal to identify either the information structure or the underlying state. If there are many signals, then it will generically be possible to recover the state by eliciting a fewer set of beliefs relative to what the Signal Priors approach requires. An interesting question is which kinds of data are best suited (both theoretically and practically) to answering a relevant question an analyst might be interested in. Along these lines, it also seems worthwhile to grapple with the practical issues related to elicitation which this paper has abstracted from (but addressed by some papers discussed above). 

Second, the use of regression in order to identify information seems worthy of further exploration. A Bayesian agent who computes an expected value, by definition, computes as a linear combination (i.e., a weighted average) of several parameters. Regression is simply a way of inverting this relationship, to find the parameters as a function of beliefs. As discussed in the literature review, this paper is not alone in noting that regression is a useful method for this reason. At the same time, by finding one application where certain regression techniques are useful, it may very well be that more elaborate regression techniques could allow for more realistic assumptions or speak to different applications and questions. This could be true even restricting to the methods used in the paper at hand, but the connection between these literatures seems worthwhile to explore further. 

\newpage

\bibliographystyle{ecta}
\bibliography{contingentbib}

\begin{thebibliography}{30}
\newcommand{\enquote}[1]{``#1''}
\expandafter\ifx\csname natexlab\endcsname\relax\def\natexlab#1{#1}\fi

\bibitem[\protect\citeauthoryear{Angeletos and Lian}{Angeletos and
  Lian}{2018}]{AngeletosLian}
\textsc{Angeletos, G.-M. and C.~Lian} (2018): \enquote{Forward Guidance without
  Common Knowledge,} \emph{American Economic Review}, 108, 2477--2512.

\bibitem[\protect\citeauthoryear{Arieli, Babichenko, and Smorodinsky}{Arieli
  et~al.}{2020}]{ABS2020}
\textsc{Arieli, I., Y.~Babichenko, and R.~Smorodinsky} (2020):
  \enquote{Identifiable information structures,} \emph{Games and Economic
  Behavior}, 120, 16--27.

\bibitem[\protect\citeauthoryear{Arieli and Mueller-Frank}{Arieli and
  Mueller-Frank}{2017}]{AM2017}
\textsc{Arieli, I. and M.~Mueller-Frank} (2017): \enquote{Inferring Beliefs
  from Actions,} \emph{Games and Economic Behavior}, 102, 455--461.

\bibitem[\protect\citeauthoryear{Aumann and Maschler}{Aumann and
  Maschler}{1995}]{AM1995}
\textsc{Aumann, R.~J. and M.~Maschler} (1995): \emph{Repeated Games with
  Incomplete Information}, MIT Press.

\bibitem[\protect\citeauthoryear{Ball}{Ball}{2021}]{Ball2021}
\textsc{Ball, I.} (2021): \enquote{Scoring Strategic Agents,} \emph{Working
  Paper}.

\bibitem[\protect\citeauthoryear{Bergemann and Morris}{Bergemann and
  Morris}{2013}]{BergemannMorris2013}
\textsc{Bergemann, D. and S.~Morris} (2013): \enquote{Robust Predictions in
  Games With Incomplete Information,} \emph{Econometrica}, 81, 1251--1308.

\bibitem[\protect\citeauthoryear{Brandts and Charness}{Brandts and
  Charness}{2011}]{BrandtsCharness2011}
\textsc{Brandts, J. and G.~Charness} (2011): \enquote{The strategy versus the
  direct-response method: a first survey of experimental comparisons,}
  \emph{Experimental Economics}, 14, 375--398.

\bibitem[\protect\citeauthoryear{Caradonna}{Caradonna}{2021}]{Caradonna2021}
\textsc{Caradonna, P.} (2021): \enquote{Preference Regression,} \emph{Working
  Paper}.

\bibitem[\protect\citeauthoryear{Chen, Mueller-Frank, and Pai}{Chen
  et~al.}{2021}]{CMP2021}
\textsc{Chen, Y.-C., M.~Mueller-Frank, and M.~Pai} (2021): \enquote{The Wisdom
  of the Crowd and Higher Order Beliefs,} \emph{Working Paper}.

\bibitem[\protect\citeauthoryear{Cripps, Ely, Mailath, and Samuelson}{Cripps
  et~al.}{2008}]{CEMS2008}
\textsc{Cripps, M., J.~Ely, G.~Mailath, and L.~Samuelson} (2008):
  \enquote{Common Learning,} \emph{Econometrica}, 76, 909--933.

\bibitem[\protect\citeauthoryear{Cvitani\'c, Prelec, Riley, and
  Tereick}{Cvitani\'c et~al.}{2019}]{CPRT2019}
\textsc{Cvitani\'c, J., D.~Prelec, B.~Riley, and B.~Tereick} (2019):
  \enquote{Honesty via Choice-Matching,} \emph{American Economic Review:
  Insights}, 1, 179--192.

\bibitem[\protect\citeauthoryear{Golub and Morris}{Golub and
  Morris}{2017}]{GolubMorris}
\textsc{Golub, B. and S.~Morris} (2017): \enquote{Higher-Order Expectations,}
  \emph{Working Paper, Northwestern University and MIT}.

\bibitem[\protect\citeauthoryear{Hellman}{Hellman}{2011}]{Hellman2011}
\textsc{Hellman, Z.} (2011): \enquote{Iterated Expectations, Compact Spaces and
  Common Priors,} \emph{Games and Economic Behavior}, 72, 163--171.

\bibitem[\protect\citeauthoryear{Horn and Johnson}{Horn and
  Johnson}{2013}]{HornJohnson}
\textsc{Horn, R.~A. and C.~R. Johnson} (2013): \emph{Matrix Analysis},
  Cambridge University Press, second ed.

\bibitem[\protect\citeauthoryear{Jakobsen}{Jakobsen}{2021}]{Jakobsen21}
\textsc{Jakobsen, A.~M.} (2021): \enquote{An Axiomatic Model of Persuasion,}
  \emph{Econometrica}, 89, 2081--2116.

\bibitem[\protect\citeauthoryear{Kamenica and Gentzkow}{Kamenica and
  Gentzkow}{2011}]{KG2011}
\textsc{Kamenica, E. and M.~Gentzkow} (2011): \enquote{Bayesian Persuasion,}
  \emph{American Economic Review}, 101, 2590--2615.

\bibitem[\protect\citeauthoryear{Lomys and Tarantino}{Lomys and
  Tarantino}{2022}]{LomysTarantino2022}
\textsc{Lomys, N. and E.~Tarantino} (2022): \enquote{Identification in Search
  Models with Social Information,} \emph{Working Paper}.

\bibitem[\protect\citeauthoryear{Lu}{Lu}{2016}]{Lu2016}
\textsc{Lu, J.} (2016): \enquote{Random Choice and Private Information,}
  \emph{Econometrica}, 84, 1983--2027.

\bibitem[\protect\citeauthoryear{Meyer}{Meyer}{1995}]{Meyer2000}
\textsc{Meyer, C.~D.} (1995): \emph{Matrix Analysis and Applied Linear
  Algebra}, SIAM.

\bibitem[\protect\citeauthoryear{Miyashita}{Miyashita}{2022}]{Miyashita2022}
\textsc{Miyashita, M.} (2022): \enquote{Identification of Information
  Structures in Bayesian Games,} \emph{Working paper}.

\bibitem[\protect\citeauthoryear{Morris and Shin}{Morris and
  Shin}{2002}]{MorrisShin}
\textsc{Morris, S. and H.~S. Shin} (2002): \enquote{Social Value of Public
  Information,} \emph{American Economic Review}, 92, 1521--1534.

\bibitem[\protect\citeauthoryear{Prelec}{Prelec}{2004}]{Prelec2004}
\textsc{Prelec, D.} (2004): \enquote{A Bayesian Truth Serum for Subjective
  Data,} \emph{Science}, 306, 462--466.

\bibitem[\protect\citeauthoryear{Prelec and McCoy}{Prelec and
  McCoy}{2022}]{PM2022}
\textsc{Prelec, D. and J.~McCoy} (2022): \enquote{General identifiability of
  possible world models for crowd wisdom.} \emph{Working Paper}.

\bibitem[\protect\citeauthoryear{Prelec, Seung, and McCoy}{Prelec
  et~al.}{2013}]{PSM2013}
\textsc{Prelec, D., H.~S. Seung, and J.~McCoy} (2013): \enquote{Finding Truth
  Even if the Crowd is Wrong,} \emph{Working Paper}.

\bibitem[\protect\citeauthoryear{Prelec, Seung, and McCoy}{Prelec
  et~al.}{2017}]{PSM2017}
---\hspace{-.1pt}---\hspace{-.1pt}--- (2017): \enquote{A solution to the
  single-question crowd wisdom problem,} \emph{Nature}, 541, 532--535.

\bibitem[\protect\citeauthoryear{Samet}{Samet}{1998{\natexlab{a}}}]{SametSeparation}
\textsc{Samet, D.} (1998{\natexlab{a}}): \enquote{Common Priors and Separation
  of Convex Sets,} \emph{Games and Economic Behavior}, 24, 172--174.

\bibitem[\protect\citeauthoryear{Samet}{Samet}{1998{\natexlab{b}}}]{SametIteratedExp}
---\hspace{-.1pt}---\hspace{-.1pt}--- (1998{\natexlab{b}}): \enquote{Iterated
  Expectations and Common Priors,} \emph{Games and Economic Behavior}, 24,
  131--141.

\bibitem[\protect\citeauthoryear{van Wieringen}{van
  Wieringen}{2015}]{RidgeTextbook}
\textsc{van Wieringen, W.~N.} (2015): \emph{Lecture Notes on Ridge Regression},
  ArXiv Preprint.

\bibitem[\protect\citeauthoryear{Witkowski and Parkes}{Witkowski and
  Parkes}{2012{\natexlab{a}}}]{WP2012b}
\textsc{Witkowski, J. and D.~C. Parkes} (2012{\natexlab{a}}): \enquote{Peer
  Prediction Without a Common prior,} \emph{Proceedings of the 13th ACM
  Conference on Electronic Commerce (EC12)}.

\bibitem[\protect\citeauthoryear{Witkowski and Parkes}{Witkowski and
  Parkes}{2012{\natexlab{b}}}]{WP2012a}
---\hspace{-.1pt}---\hspace{-.1pt}--- (2012{\natexlab{b}}): \enquote{A Robust
  Bayesian Truth Serum for Small Populations,} \emph{Proceedings of the 26th
  AAAI Conference on Artificial Intelligence (AAAI 13)}.

\end{thebibliography}

\newpage

\appendix 

\section{Proofs of the Main Result}

\begin{proof}[Proof of Theorem \ref{thm:identification}] 

The first part of the theorem follows the arguments as laid in the main text. Writing the information structure so that rows are states and columns are signals, the definition of hypothetical beliefs matrix tells us that: 

\begin{equation*} 
Q=B \cdot \mathcal{I}, 
\end{equation*}

\noindent since the $i$th row and $j$th column of $Q$ is the inner product of the $i$th row of $B$ (since rows of $B$ index signals) and the $j$th column of $\mathcal{I}$ (using the convention that columns index signals). Given this expression, and using that $(B^{T}B)^{-1}$ is invertible, the solution for $\mathcal{I}$ comes from left-multiplying both sides by $B^{T}$ and then left multiplying by $(B^{T}B)^{-1}$. 

Next, I show that the prior is identified. As shown in the main text, the prior is a unit eigenvector (with eigenvalue 1) of the matrix $B^{T} \mathcal{I}^{T}$, and therefore a unit eigenvector of $B^{T} Q^{T}B (B^{T}B)^{-1}$. I show that this matrix is \emph{always} row-stochastic. Recall that $\mathbf{1}^{T}B^{T}$ and $\mathbf{1}^{T}Q^{T}$ are both $\mathbf{1}$, since both $B$ and $Q$ are row-stochastic (so that the transposes are column-stochastic). Therefore:

\begin{equation*} 
\mathbf{1}^{T} B^{T}Q^{T}B(B^{T}B)^{-1}= \mathbf{1}^{T} Q^{T} B (B^{T}B)^{-1}= \mathbf{1}^{T} B(B^{T}B)^{-1}
\end{equation*}

\noindent Taking the transpose of this expression gives: 

\begin{equation*} 
(B^{T}B)^{-1}B^{T} \mathbf{1}.
\end{equation*}

\noindent Now, recall the regression interpretation of the linear mapping $(B^{T}B)^{-1}B^{T}$; when applied to a vector, it gives the coefficient of the regression of the vector onto the columns of $B$. However, the columns of $B$ sum to 1. Therefore, to write the vector $\mathbf{1}$ as a linear combination of the columns of $B$, I need only write it using coefficients equal to 1, and hence this expression is itself a vector of 1s, demonstrate that the matrix is row-stochastic. 

Hence, the Frobenius-Perron theorem holds provided the matrix $B^{T} \mathcal{I}^{T}$ is irreducible. This property holds generically; indeed, the entries of $B^{T} \mathcal{I}^{T}$ are generically positive, and all such matrices are irreducible. This theorem therefore yields the existence of a Perron eigenvector, which is positive and sums to 1. While the argument in the main text shows that being a Perron eigenvector is a \emph{necessary} condition for the prior, the Frobenius-Perron theorem implies that this vector is unique, and therefore this condition is also sufficient. As a result, the prior is additionally identified from $B$ and $Q$, in addition to the information structure, as desired. 
\end{proof}

\begin{proof}[Proof of Proposition \ref{prop:dimproof}] 

By Proposition 1 of \cite{KG2011}, given any belief matrix $B$ and prior $p$, there exists an information structure $\mathcal{I}$ inducing this belief matrix.\footnote{Their proof is constructive; in my notation, one can set $\mathcal{I}(\theta)[s]=B_{s,\theta}\Prob[s]/p[\theta]$. Now, note that $\Prob[s]$ is a left unit eigenvector of the matrix $Q$. Note, however, that $\Prob[s]$ would be derived from $Q$ in this paper, and not the prior $p$ and the posterior beliefs, as in theirs.} On the other hand, Theorem \ref{thm:identification} shows that \emph{any} vector $v$ of length $\abs{S}$ yields a vector of length $\Theta$ when considering $(B^{T}B)^{-1}B^{T} v$. Thus, the set of information structures is spanned by the set of $Q$ that emerge in some informational environment. Putting the previous observations together, the set of $Q$ which induce an informational environment given a belief matrix $B$ is isomorphic to the set of priors inducing $B$, which has dimensionality equal to $\abs{\Theta}-1$, for any belief matrix satisfying the linear independence condition. 
\end{proof}

\begin{proof}[Proof of Proposition \ref{prop:dimQ}]
It is immediate that the set of row stochastic matrices has dimension $\abs{S}(\abs{S}-1)$. I show that this is also the dimensionality of the set of plausible matrices which induce valid information structures. 

I first argue that the rows of $(B^{T}B)^{-1}B^{T}Q$ sum to 1 if $Q$ has rows that sum to 1. I note that the sum of the rows of this matrix is given by right multiplying by $\mathbf{1}_{\abs{S}}$, a vector of length $\abs{S}$ which is all 1s. If $Q$ has rows which sum to 1, then $(B^{T}B)^{-1}B^{T}Q\cdot \mathbf{1}_{\abs{S}}=(B^{T}B)^{-1}B^{T}\mathbf{1}_{\abs{S}}$. On the other hand, recall that this expression is also the coefficients $\beta_{1}, \ldots, \beta_{n}$ solving: 

\begin{equation*} 
\mathbf{1}_{\abs{S}}= \sum_{i=1}^{n} \beta_{i} b_{i}, 
\end{equation*} 

\noindent where $b_{i}$ is the $i$th column of the belief matrix $B_{i}$. While there is a unique set of coefficients solving this equation, I also have that the columns of a belief matrix sum to 1. Hence $\beta_{1}= \cdots = \beta_{n}=1$ is the solution. I thus conclude that $(B^{T}B)^{-1}B^{T}Q\cdot \mathbf{1}_{\abs{S}}$ is a vector of 1s, as claimed. 

Now, consider an arbitrary information structure $\tilde{\mathcal{I}}$ with \emph{only strictly positive entires}, and let $Q:=B^{T} \tilde{\mathcal{I}}$. Then considering the regression coefficients as in Theorem \ref{thm:identification}, I have $\tilde{\mathcal{I}}= (B^{T}B)^{-1}B^{T}Q$, which is a valid information structure. 

Furthermore, given that $B$ is full rank, $(B^{T}B)^{-1}B^{T}$ is a continuous transformation on the set of row stochastic matrices. Since $(B^{T}B)^{-1}B^{T}Q$ is strictly positive, by the definition of continuity, I have there is an open subset (assuming the relative topology) of row-stochastic matrices $\tilde{Q}$ such that $(B^{T}B)^{-1}B^{T}Q$ is strictly positive. This implies that the set of plausible $Q$ has the same dimensionality as the set of row-stochastic $Q$, since an open set within a topological space has the same dimensionality as the space. This completes the proof.   \end{proof}

\begin{proof}[Proof of Proposition \ref{prop:noncommon}]
Note that, given any $B$ satisfying the assumptions of my framework, that $B^{T}B$ is invertable, so that $(B^{T}B)^{-1}B^{T}Q$ is well-defined. I consider an information environment where all belief types use this information structure, but each belief type arrives at posterior beliefs by updating an belief-type specific prior $p^{s}$. I write $(b_{s,\theta})_{\theta}$ as the belief vector of type for belief type $s$, and consider some fixed $s$. Choose $r(\theta)$ to satisfy:

\begin{equation*} 
r(\theta)\mathcal{I}(\theta)[s] = b_{s,\theta},
\end{equation*}

\noindent noting that, since all other terms in this expression are non-negative, $r(\theta)$ is as well. Further note that one cannot have $\mathcal{I}(\theta)[s]=b_{s,\theta}=0$ for all $\theta$, since $\sum_{\theta} b_{s,\theta}=1$ by assumption. Hence $r(\theta) \geq 0$ for all $\theta$ and $r(\theta) \neq 0$ for some $\theta$. Therefore, I can set: 

\begin{equation*} 
p_{s}(\theta) =  \frac{r(\theta)}{\sum_{\tilde{\theta}} r(\tilde{\theta})}, 
\end{equation*} 

\noindent which is an element of $\Delta(\Theta)$. 

Suppose belief type $s$ updates using $\mathcal{I}$ and prior $p^{s}$. In that case, following signal $s$, the Bayesian belief is: 

\begin{equation*} 
\frac{ \frac{r(\theta)}{\sum_{\tilde{\theta}} r(\tilde{\theta})} \mathcal{I}(\theta)[s]}{\sum_{\theta'} \frac{r(\theta')}{\sum_{\tilde{\theta}} r(\tilde{\theta})} \mathcal{I}(\theta')[s]} = \frac{r(\theta) \mathcal{I}(\theta)[s]}{\sum_{\theta'} r(\theta') \mathcal{I}(\theta')[s]}  = \frac{b_{s,\theta}}{\sum_{\theta'}b_{s,\theta'}}=b_{s,\theta},
\end{equation*}

\noindent where all equations follow from either cancellations or definitions. Furthermore, the probability belief type $s$ assigns to a randomly selected individual having belief type $\tilde{s}$ is $q_{s,\tilde{s}}= \sum_{\theta} \mathcal{I}(\theta)[\tilde{s}]b_{s,\theta}$, since posterior beliefs are $b_{s,\theta}$ and the assumed information structure is $\mathcal{I}(\theta)$. This completes the proof.  \end{proof}

\begin{proof}[Proof of Theorem  \ref{thm:morestates}]
First, I claim that the limit $\lim_{\lambda \rightarrow 0} (B^{T}B + \lambda I)^{-1} B^{T}Q$ exists and is a finite matrix; I provide an independent proof of this result after completing the rest of the proof, though I mention that this fact appears known (though requiring some additional detours the proof below avoids).\footnote{As per \cite{RidgeTextbook}, the limit of the ridge estimator as $\lambda \rightarrow 0$ is precisely the least square estimate of smallest norm; showing this, however, requires a significant detour in defining ridge estimators. Note that \cite{RidgeTextbook} also shows that multiplying by a matrix $M$ as described in the main text amounts to a rescaling of the design matrix (in this case, $B$).}

I now turn to the second bulletpoint of the Theorem. Note that the ridge estimator defined by $\tilde{\mathcal{I}}$ solves the following minimization problem: 

\begin{equation} 
\tilde{\mathcal{I}}_{\lambda}(\cdot)[s]= \arg \hspace{-.70mm} \min_{\hspace{-5.5mm} x} ~~ \abs{\abs{q_{s, \cdot} - Bx}}^{2} + \lambda \abs{\abs{x}}^{2}. \label{eq:objective}
\end{equation}

By contrast, the information structure $\mathcal{I}$ solves $q_{s, \cdot}=B\mathcal{I}(\cdot)[s]$. Now, given the claim that $\lim_{\lambda \rightarrow 0}\tilde{\mathcal{I}}_{\lambda}(\cdot)[s]$ exists, it follows from this expression that the resulting limit must also be a solution to the equation $q_{s, \cdot}=Bx$; if it weren't, then one would have the objective in (\ref{eq:objective}) would converge to some strictly positive amount as $\lambda \rightarrow 0$; by contrast, \emph{any} solution to this equation makes this objective equal to 0 in the limit.  Hence any vector $x$ which does not satisfy $q_{s, \cdot}=Bx$ cannot be the limit of $\tilde{\mathcal{I}}_{\lambda}(\cdot)[s]$ as $\lambda \rightarrow 0$. 

On the other hand, for \emph{any} vector $x$ satisfying $q_{s, \cdot} = Bx$, subtracting the equation for $\tilde{\mathcal{I}}$ from this equation yields $0=B(x-\tilde{\mathcal{I}}(\cdot)[s])$, so that $x-\tilde{\mathcal{I}}(\cdot)[s]$ is in the nullspace of $B$. But the information structure generating the decision-maker's information is one possible choice of $x$; therefore, $\mathcal{I}(\cdot)[s]-\tilde{\mathcal{I}}(\cdot)[s]=\sum_{i} \alpha_{i} v_{i}$, where $\{v_{1}, \ldots, v_{k}\}$ is a basis for the nullspace of $B$ (assuming the dimension of this space is $k$); adding $\tilde{\mathcal{I}}(\cdot)[s]$ to both sides of this expression proves the second bulletpoint.\footnote{An identical argument shows the claim that \emph{every} solution to the equation $Q=BX$ differs from $\mathcal{I}$ in this way.}

Next, I show that $p$ is a unit eigenvector with eigenvalue 1. Suppose $Q=B \tilde{\mathcal{I}}$, and let $q$ denote left unit eigenvector with eigenvalue 1. Recall that (see Section \ref{sect:SP}) that $q=(\mathbb{P}[s])_{s \in S}$, the vector of probabilities that each signal in $s$ is observed given the informational environment. However, note that the martingale properties of beliefs states that $\mathbb{P}[s] \cdot B=p(\theta)$; putting this together gives us that, for every $s$, I have $\sum_{\theta} \tilde{\mathcal{I}}(\theta)[s] \cdot p(\theta) = q_{s}$ (i.e., the entry of $q$ corresponding to $s$. On the other hand, the martingale property of beliefs, written $qB=p$ (where $p$ is the prior), does not depend on the information structure. Thus, I can apply the same argument to say that the prior is a left unit eigenvector (with eigenvalue 1) of $\tilde{\mathcal{I}}B$. Note that this does not complete the proof since I still have to show this eigenvector is unique. 

Letting $p_{X}(\lambda)$ be the characteristic polynomial for a matrix $X$, I note that for matrices $A$ and $B$ where $A$ is $m$-by-$n$ and $B$ is $n$-by-$m$, with $n \geq m$ one has $p_{BA}(\lambda)=\lambda^{n-m}p_{AB}(\lambda)$. I apply this result to $B^{T} \mathcal{I}^{T}$ and $B^{T} \tilde{\mathcal{I}}^{T}$ (see Theorem 1.3.22 in \cite{HornJohnson}). In particular, using the previous result, write $\tilde{\mathcal{I}}$ as $\mathcal{I}+W$, where each row of $W$ is in the null space of $B$. In particular, since each $W$ is in the null space of $B$, I have: 

\begin{equation*} 
(\mathcal{I}+W)^{T}B^{T} = \mathcal{I}^{T} B^{T}. 
\end{equation*}

\noindent Putting this together with the previous results, using that $\abs{\Theta} > \abs{S}$ (so that $B^{T}$ has more rows than columns), I have: 

\begin{equation*} 
p_{B^{T} \tilde{\mathcal{I}}^{T}}(\lambda)= \lambda^{\abs{\Theta}-\abs{S}}
p_{\tilde{\mathcal{I}}^{T}B^{T} }(\lambda)=\lambda^{\abs{\Theta}-\abs{S}}
p_{\tilde{\mathcal{I}}^{T}B^{T} }(\lambda)=p_{B^{T} \mathcal{I}^{T}}(\lambda)
\end{equation*}

\noindent As argued in Theorem \ref{thm:identification}, $B^{T}\mathcal{I}^{T}$ has a unique eigenvector with eignvalue 1 (i.e., there is no multiplicity in the eigenspace). Thus, since $B^{T} \mathcal{I}^{T}$ has a unique eigenvector with eigenvalue 1, so does $B^{T}\tilde{\mathcal{I}}$; as argued above, this must be the prior, completing the proof. 

I conclude by showing that $\lim_{\lambda \rightarrow 0} (B^{T}B + \lambda I)^{-1} B^{T}Q$ exists and is a finite matrix, as claimed. I first determine the rate at which the determinant of $B^{T}B+\lambda I$ tends to 0 as $\lambda \rightarrow 0$. Note that, by the Matrix Determinant Lemma (see (6.2.3) of \cite{Meyer2000} for a version of this result), I have: 

\begin{equation*}  
\det(\frac{1}{\lambda} B^{T}B + I) = \det(I_{\abs{S}}+ \frac{1}{\lambda}BB^{T}).
\end{equation*} 

\noindent Note that the matrix involved in the left-hand side of this equation is $\abs{\Theta}-by-\abs{\Theta}$ and the matrix involved in the right hand side of this equation is $\abs{S}-by-\abs{S}$. I therefore have, multiplying through by $\lambda^{\abs{\Theta}}$ and using that $\det(cA)=c^{n}\det(A)$ for $c \in \mathbb{R}$ and $A$ an $n$-by-$n$ matrix,

\begin{equation*} 
\det( B^{T}B + \lambda I) = \det( \lambda^{\abs{\Theta}/\abs{S}} I_{\abs{S}}+ \lambda^{(\abs{\Theta}- \abs{S})/\abs{S}}BB^{T}).
\end{equation*}

\noindent Note that this determinant is a polynomial in $\lambda$ which evaluates to 0 at $\lambda=0$, and hence this approaches 0 at a rate equal to the rate of the smallest term in this polynomial. I claim the degree is strictly less than $\abs{\Theta}$. This is clear from examining the right hand side of the equation above. While every term on the diagonal in this matrix is of the order $\lambda^{\abs{\Theta}/\abs{S}}$, every term \emph{off} the diagonal is of the order $\lambda^{\abs{\Theta}/\abs{S}-1}$. To show that there is a term in the polynomial defined by the determinant that is of order less than $\abs{\Theta}$, it suffices to show that some term in this expression reflects off diagonal terms. Note that the determinant is a sum over permutations $\sigma : \abs{S} \rightarrow \abs{S}$: 

\begin{equation*} 
\sum_{\sigma}\text{sgn}(\sigma) \prod_{i=1}^{n} b_{i, \cdot} b_{\cdot, \sigma(i)},
\end{equation*}

\noindent where $b_{i, \cdot}$ is the $i$th row of $B$. Then the permutations which simply swap two elements (of which there are $\abs{S}\cdot (\abs{S}-1)/2$ of) contribute to the determinant; for any such permutation, $\text{sgn}(\sigma)=-1$. Since any coefficient in this sum where $\lambda$ is of degree $\abs{\Theta}-2$ must correspond to one of these permutations, the exponent on $\lambda$ reflecting these permutations is at most $\abs{\Theta}-2$; and thus, the smallest non-zero degree of the characterisic polynomial must be less than $\abs{\Theta}$.\footnote{More generally, the lowest degree of the polynomial should be $\abs{\Theta}- \abs{S}$; showing this, however, requires that some permutations which influence the determinant do not fix any elements on the diagonal. Determining that not all terms cancel out, while certainly intuitive, appears less direct than this argument. However, provided this is the case, then any entry corresponding to exclusively off-diagonal term will be a polynomial of order $(\lambda^{(\abs{\Theta}-\abs{S})/\abs{S}})^{\abs{S}}$, since the matrix is $\abs{S}$-by-$\abs{S}$.}

Therefore, as $\lambda \rightarrow 0$, $\det \left((B^{T}B + \lambda I) \frac{1}{\lambda}I \right)= \det(B^{T}B + \lambda I) \frac{1}{\lambda^{\abs{\Theta}}} \not\rightarrow 0$. Taking inverses, $(B^{T}B + \lambda I)^{-1} \lambda I$ must have a limit; indeed, in the definition of the matrix inverse, each term is scaled by the inverse of the determinant, and otherwise comes from multiplying and adding matrix elements together---so, since each term is scaled by a term that does not approach infinity, each term converges to a finite limit. Using Equation \ref{eq:ridgeequation}, I conclude that the limit defining $\tilde{\mathcal{I}}$ exists.\end{proof}

\section{Additional Results and Discussion}
\subsection{Removing Linear Dependencies} \label{app:lineardependence}
This section comments on the full rank assumptions. First, note that this assumption is generic: Indeed, the set of all possible belief matrices is of dimension $\abs{S} \times (\abs{\Theta}-1)$ (with one degree of freedom lost for every row, since every row is restricted to sum to 1); in general, the space spanned by $k$ belief vectors, restricted to sum to 1, is $(k-1) \times \abs{S}$ dimensional. Hence a generic belief matrices is full rank, in that the set of belief matrices which fail to satisfy this belong to a lower dimension subspace.  Given this observation, $B$ is generically of rank equal to $\abs{\Theta}$. On the other hand, the rank of $B^{T}B$ is equal to the rank of $B$, and therefore $B^{T}B$ has full rank. Since a square matrix is invertible if and only if it has full rank, I have that $B^{T}B$ is invertible.

I now comment on the implicit assumption behind the full rank assumption that the columns of $B$ are linearly independent. While the argument in the previous paragraph shows this assumption is generic, one may be interested in cases in where it is violated or understanding the substance of this assumption. Note that if this condition fails, then the matrix $B^{T}B$ is not invertible.\footnote{Indeed, as discussed in Section \ref{sect:morestates}, this condition always \emph{fails} when $\abs{S} > \abs{\Theta}$.} 

I show how the case of linearly dependent columns can be interpreted as reflecting the case where a state is ``split apart.'' Rather than discussing this in full generality, I present an illustrative example. 

Suppose that $\Theta= \{\theta_{1}, \theta_{2}, \theta_{3}, \theta_{4}\}$, where the third column of $B$ is a linear combination of the first two: 

\begin{equation*} 
B=\begin{pmatrix} 
2/3 & 0 & 1/3 & 0 \\ 1/3 & 1/3 & 1/3 & 0 \\ 0 & 2/5 & 1/5 & 2/5 \\ 0 & 0 & 0 & 1 
\end{pmatrix}, ~~~ Q=\begin{pmatrix} 
1/2 & 1/2 & 0 & 0 \\ 1/4 & 1/2 & 1/4 & 0 \\ 0 & 3/10 & 1/2 & 1/5 \\ 0 & 0 & 1/2 & 1/2
\end{pmatrix}
\end{equation*}
\noindent Even though $B$ is 4-by-4, $B^{T}B$ is not invertible, as the linear independence condition is not satisfied; specifically, the third column is $1/2$ times the first column and $1/2$ times the second column. Ideally, one could ``remove'' the third state responsible for the linear dependencies. Importantly, since the hypothetical belief matrix makes no reference to the underlying states, it would not change if states were removed, provided the distributions over the signals were to not change. 

I now show how to remove the state $\theta_{3}$, and subsequently interpret the original state space as an auxiliary one where $\theta_{3}$ is induced with equal probabilities following $\theta_{1}$ and $\theta_{2}$ (\emph{after} these are already drawn). After doing this, it will be possible to recover the information structure and prior. Renormalize $B$ so that it does not include $\theta_{3}$; that is, consider the belief matrix that would emerge conditional on $\{\theta_{1}, \theta_{2}, \theta_{4}\}$. Considering the belief matrix that emerges when I remove $\theta_{3}$ in this way, I have: 

\begin{equation*}
\tilde{B}=\begin{pmatrix} 
1 & 0  & 0 \\ 1/2 & 1/2  & 0 \\ 0 & 3/5 & 2/5 \\ 0 & 0 & 1
\end{pmatrix}.
\end{equation*}

\noindent If one were to have started with $\tilde{B}$, then $B$ could be obtained by considering the case where the state is ``flipped'' to $\theta_{3}$ following $\theta_{1}$, with probability $1/3$, and the state is ``flipped'' to $\theta_{3}$ following $\theta_{2}$, again with probability $1/3$ (and never following $\theta_{4}$). Indeed, the third column of $B$ is the sum of the first two columns, times $1/2$; and the first two columns of $B$ are the same as the first two columns of $\tilde{B}$, divided by $2/3$ (and $2/3$ is the probability that the state is ``unflipped''). This is the sense in which $\theta_{3}$ is a linear combination of $\theta_{1}$ and $\theta_{2}$. In particular: whenever one column is a convex combination of other columns, one can simply eliminate it from the belief matrix, and then ``regenerate'' it in this way. 

Now, the matrix $\tilde{B}^{T} \tilde{B}$ \emph{is} invertible, and regressing each column of $Q$ on $\tilde{B}$ gives an information structure. In this case: 

\begin{equation*} 
(\tilde{B}^{T} \tilde{B})^{-1} \tilde{B} Q= \begin{pmatrix} 1/2 & 1/2 & 0 & 0 \\ 0 & 1/2 & 1/2 & 0 \\ 0 & 0 & 1/2 & 1/2  \end{pmatrix}.
\end{equation*}

\noindent One can check that this information structure generates $\tilde{B}$ and $Q$, as Theorem \ref{thm:identification} suggests it should, using the prior $\Prob[\theta_{1}]=\Prob[\theta_{2}]=3/8$ and $\Prob[\theta_{4}]=1/4$.

Now, notice that in the above interpretation, $\theta_{3}$ is induced with equal probabilities following $\theta_{1}$ and $\theta_{2}$. So consider the following information structure on the original state space $\{\theta_{1}, \theta_{2}, \theta_{3}, \theta_{4} \}$:

\begin{equation*} 
\mathcal{I}=\begin{pmatrix} 
1/2 & 1/2 & 0 & 0 \\ 0 & 1/2 & 1/2& 0 \\ 1/4 & 1/2 & 1/4 & 0 \\ 0 & 0 & 1/2 & 1/2
\end{pmatrix}.
\end{equation*}

\noindent  Where does the signal distribution following state $\theta_{3}$ come from? The third row of this vector is one half the first row plus one half the second row. In other words, the signal distribution is exactly what it would be if ``the state is $\theta_{3}$'' is equivalent to ``the state is $\theta_{1}$ with probability 1/2 and $\theta_{2}$ with probability 1/2.'' And indeed, one can check that this information structure, under a uniform prior (which, again, is what would the prior would be under the specification of how $\theta_{3}$ is determined from $\theta_{1}$ and $\theta_{2}$), generates $B$. 

\subsection{Deterministically Generated Signals} \label{app:irreducibility}

In this section I discuss a special case of the framework, where signals are generated deterministically as a function of the state. There are three reasons why this is of interest. One is simply practical---there are cases where a decisionmaker may observe a partition of the state space, and this alternative describes this model. Second, it is likely the simplest case where $B^{T} \mathcal{I}^{T}$ fails to be irreducible, allowing me to illustrate that the prior is not identified (and in particular why one should not expect the prior to be identified, and what is identified instead). Third, this case imposes that $\abs{\Theta} > \abs{S}$, since a partition of a state space by definition cannot have more elements than a state space.  We will see that the ridge regression procedure can, in this case, produce an information structure, albeit one using a different limiting equation. 

Suppose $\Theta=\{\theta_{1}, \theta_{2}, \theta_{3}, \theta_{4}\}$, and consider an information structure where the decision-maker observes which element of $\mathcal{P}= \{\{\theta_{1}\}, \{\theta_{2}, \theta_{3} \}, \{\theta_{4} \}\}$ the state belongs to. If $p(\theta_{i})$ is the prior probability over state $\theta_{i}$, then this corresponds to the following state belief matrix and hypothetical belief matrix: 

\begin{equation*} 
B= \begin{pmatrix} 1 & 0 & 0 & 0 \\ 0 & \frac{p(\theta_{2})}{p(\theta_{2})+p(\theta_{3})} & \frac{p(\theta_{3})}{p(\theta_{2}) + p(\theta_{3})} & 0 \\ 0 & 0 & 0 & 1 \end{pmatrix}, ~~~ Q = \begin{pmatrix} 1 & 0 & 0 \\ 0 & 1 & 0 \\ 0 & 0 & 1 \end{pmatrix}
\end{equation*}

\noindent  Of course, in this example my main exercise is fairly straightforward, but this directness will be helpful in understanding the implications of irreducibility of $B^{T}\mathcal{I}^{T}$ as well as why the ridge estimator does not yield the correct information structure.  I compute: 

\begin{equation*} 
\lim_{\lambda \rightarrow 0} ~~ (B^{T}B + \lambda I)^{-1} \lambda I = \begin{pmatrix} 0 & 0 & 0 \\ 0 & \frac{p(\theta_{3})(p(\theta_{3})-p(\theta_{2}))}{p(\theta_{2})^{2} + p(\theta_{3})^{2}} & 0 \\ 0 & \frac{p(\theta_{2})(p(\theta_{2})-p(\theta_{3}))}{p(\theta_{2})^{2}+p(\theta_{3})^{2}} & 0 \\ 0 & 0 & 0 \end{pmatrix}
\end{equation*}

This matrix converges to 0 only when $p(\theta_{2}) = p(\theta_{3})$. And indeed, the procedure fails to produce an information structure:  

\begin{equation*} 
\lim_{\lambda \rightarrow 0} ~~ (B^{T}B + \lambda I)^{-1} B^{T}Q = \begin{pmatrix} 1 & 0 & 0 \\ 0 & \frac{p(\theta_{2})(p(\theta_{2})+p(\theta_{3}))}{p(\theta_{2})^{2}+ p(\theta_{3})^{2}} & 0 \\ 0 & \frac{p(\theta_{3})(p(\theta_{2})+p(\theta_{3}))}{p(\theta_{2})^{2}+ p(\theta_{3})^{2}} & 0 \\ 0 & 0 & 1 \end{pmatrix}.
\end{equation*}

\noindent What went wrong? Notice that in the above derivation, adding $\lambda I$ to $B^{T}B$ was only one way to ensure the inversion step would be possible. In this case, the procedure arrives at a solution for $Q=B \tilde{\mathcal{I}}$ which does not correspond to the true information structure. As discussed, this equation \emph{should} have multiple solutions in the case of $\abs{\Theta} > \abs{S}$, while the limit only considers one of them. 

Note that in this case, a different regularization would deliver the information structure. For instance, one can compute that: 

\begin{equation*} 
\lim_{\lambda \rightarrow 0} \left(B^{T}B + \lambda \overbrace{\begin{pmatrix} 1 & 0 & 0 & 0 \\ 0 & 1 & 0 & 0 \\ 0 & 0 & \frac{p(\theta_{3})}{p(\theta_{2})} & 0 \\ 0 & 0 & 0 & 1 \end{pmatrix}}^{:= M} \right)^{-1}B^{T}Q = \begin{pmatrix} 1 & 0 & 0 \\ 0 & 1 & 0 \\ 0 & 1 & 0 \\ 0 & 0 & 1 \end{pmatrix},
\end{equation*}

\noindent which is indeed the deterministic information structure in this example. One can derive this expression by following the same steps as outlined in Equation (\ref{eq:ridgeequation}), but considering a different perturbation; namely, using $\lambda M$ instead of $\lambda I$.

And indeed, there is no unique prior identified by $B$ and $Q$ in this case---while one can determine the prior beliefs ``within a signal,'' in the deterministic information structure case, it is not possible to determine the relative probability across different partition elements. More generally, as illustrated by this example, if $B^{T}\mathcal{I}^{T}$ is not irreducible, then while one can use the Frobenius-Perron theorem to determine a unique prior within each irreducible class, one cannot determine the \emph{relative} prior probabilities across classes. These relative probabilities cannot possibly be determined using \emph{any} ex-post data, reflecting the substantive implications of the failure of irreducibility.

To conclude this section, I note that an information structure is deterministic if and only if $Q$ is the identity. Note that in this case, $ \abs{\Theta} > \abs{S}$ whenever the information structure does not reveal the state. 

A deterministic information structure involves the decision-maker observing an element of the partition of $\Theta$. Suppose an information structure is partitional. This implies that the probability of observing any signal given any state is either 0 or 1. On the other hand, $b_{s, \theta}$ is positive if and only if $\mathcal{I}(\theta)[s] =1$, meaning that $q_{s, \tilde{s}}$ is equal to 1 if $s = \tilde{s}$ and 0 otherwise. Therefore, $Q$ is the identity. 

Now suppose $Q$ is the identity. Notice that each entry of $Q$ is a convex combination of $\mathcal{I}(\theta)[\cdot]$, weighted according to a row of $B$. So, if $b_{s, \theta} > 0$, then I must have $\mathcal{I}(\theta)[s]=1$. Notice that this immediately implies that this partitions the state space, since one cannot have two signals $s, s'$ for which $b_{s, \theta} > 0$, since this would imply the rows of $\mathcal{I}$ sum to a number greater than 1. Therefore, I obtain a partition of a subset of the state space $\Theta$; for any $\theta \in \Theta$ that is not in this subset, I have $b_{s, \theta}=0$ for all $s \in S$. In this case, $B$ and $Q$ are generated according to a partitional information structure, where each element of the partition is the support of $b_{s, \theta}$ for some $s$, and where the prior assigns probability 0 to any state where $b_{s, \theta}=0$ for all $s$. Hence, the information is generated by a partition. 

\end{document}